\documentclass[a4paper,reqno]{amsart}
\usepackage[all]{xy} 
\usepackage{amssymb} 
\usepackage{hyperref}
\usepackage{eucal}
\usepackage{graphicx}
\usepackage{epsfig}
\usepackage{xcolor}

  \vsize
27.9truecm \hsize 21.59truecm \hoffset -0.4truecm
\numberwithin{equation}{section}

\newtheorem{definition}{Definition}[section]

\newtheorem{theorem}[definition]{Theorem}
\newtheorem{proposition}[definition]{Proposition}
\newtheorem{corollary}[definition]{Corollary}
\newtheorem{remarkth}[definition]{Remark}
\newtheorem{example}[definition]{Example}
\newenvironment{remark}{\begin{remarkth}\upshape}{\hfill$\diamond$\end{remarkth}}

\renewcommand{\emph}[1]{{\bfseries\itshape{#1}}}

\newcommand{\R}{\mathbb{R}} 
\newcommand{\N}{\mathbb{N}} 
\newcommand{\C}{\mathcal{C}}

\newcommand{\lcf}{\lbrack\! \lbrack}
\newcommand{\rcf}{\rbrack\! \rbrack}

\makeatletter
\newcommand\prol{\@ifstar{\@proldf}{\@prolpf}} 
\def\@prolpf{\@ifnextchar[{\@prolpf@wrt}{\@prolpf@}}
\def\@prolpf@wrt[#1]#2{\@ifnextchar[{\@prolpf@wrt@at{#1}{#2}}{\@prolpf@wrt@{#1}{#2}}}

\def\@prolpf@wrt@at#1#2[#3]{\prolsymbol^{#1}_{#3}#2}
\def\@prolpf@wrt@#1#2{\prolsymbol^{#1}#2}
\def\@prolpf@#1{\@ifnextchar[{\@prolpf@at{#1}}{\@prolpf@@{#1}}}
\def\@prolpf@at#1[#2]{\prolsymbol_{#2}#1}
\def\@prolpf@@#1{\prolsymbol#1}
\def\@proldf{\@ifnextchar[{\@proldf@wrt}{\@proldf@}}
\def\@proldf@wrt[#1]#2{\@ifnextchar[{\@proldf@wrt@at{#1}{#2}}{\@proldf@wrt@{#1}{#2}}}

\def\@proldf@wrt@at#1#2[#3]{\prolsymbol^{*#1}_{#3}#2}
\def\@proldf@wrt@#1#2{\prolsymbol^{*#1}#2}
\def\@proldf@#1{\@ifnextchar[{\@proldf@at{#1}}{\@proldf@@{#1}}}
\def\@proldf@at#1[#2]{\prolsymbol^*_{#2}#1}
\def\@proldf@@#1{\prolsymbol^*#1}
\def\prolsymbol{\mathcal{L}}
\makeatother

\newcommand{\Gc}{\mathcal{G}}

\def\lcf{\lbrack\! \lbrack}
\def\rcf{\rbrack\! \rbrack}
\begin{document}
\title[A unified framework for mechanics]
{A unified framework for mechanics.\\ Hamilton-Jacobi equation and applications}

\author[P.\ Balseiro]{P.\ Balseiro}
\address{P.\ Balseiro:
Instituto de Matematica Pura e Aplicada (IMPA), Estrada Dona Castorina 110, 22460-320 Rio de Janeiro, Brazil } \email{poi@impa.br}

\author[J.\ C.\ Marrero]{J.\ C.\ Marrero}
\address{J.\ C.\ Marrero:
Unidad Asociada ULL-CSIC Geometr{\'\i}a Diferencial y Mec\'anica
Geom\'etrica, Departamento de Matem\'atica Fundamental, Facultad de
Matem\'aticas, Universidad de la Laguna, La Laguna, Tenerife, Canary
Islands, Spain} \email{jcmarrer@ull.es}

\author[D.\ Mart\'{\i}n de Diego]{D. Mart\'{\i}n de Diego}
\address{D.\ Mart\'{\i}n de Diego:
Instituto de Ciencias Matem\'aticas (CSIC-UAM-UC3M-UCM), Serrano
123, 28006 Madrid, Spain} \email{d.martin@imaff.cfmac.csic.es}

\author[E.\ Padr\'on]{E.\ Padr\'on}
\address{E.\ Padr\'on: Unidad Asociada ULL-CSIC Geometr{\'i}a Diferencial
y Mec\'anica Geom\'etrica, Departamento de Matem\'atica Fundamental,
Facultad de Matem\'aticas, Universidad de la Laguna, La Laguna,
Tenerife, Canary Islands, Spain} \email{mepadron@ull.es}

\keywords{skew-symmetric algebroids, 1-cocycle, linear Poisson
structures, AV-bundles, Hamiltonian dynamics, Hamilton-Jacobi
equation, affine nonholonomic systems, linear external forces.}

\subjclass[2000]{70H20,70G45,70F25,37J60. }

\thanks{This work has been partially supported by MEC(Spain) Grants MTM2006-03322, MTM2007-62478, MTM2009-13383, ACIISI SOLSUBC200801000238, project
ÒIngenioMathematicaÓ(i-MATH) No. CSD2006-00032
(Consolider-Ingenio2010) and S-0505/ESP/0158 of  the CAM. \ \
P.B. thanks CNPq(Brazil) for financial support and the Centre
Interfacultaire Bernoulli (Switzerland) for its hospitality during
the Program {\it Advances in the Theory of Control Signals and
Systems with Physical Modeling}, where part of this work was done.
The authors wish to thank Jos\'e M\'endez for helpful comments about
the examples contained in this paper.  }


\begin{abstract}
In this paper, we construct Hamilton-Jacobi equations for a great
variety of mechanical systems (nonholonomic systems subjected to linear or
affine constraints, dissipative systems subjected to external forces, time-dependent mechanical
systems...). We recover all these, in principle, different cases
using a unified framework based on skew-symmetric algebroids with a
distinguished $1$-cocycle. Several examples illustrate the theory.

\end{abstract}

\maketitle
\tableofcontents

\section{Introduction}

A fundamental requirement for new developments in mechanics is to
unravel the geometry that underlies different dynamical systems,
especially mechanical systems.  There are several reasons why this
geometrical understanding is fundamental. First, it is a key tool
for reduction by symmetries and for the geometric characterization
of the integrability and stability theories. Second, the effective
use of numerical techniques is often based on the comprehension of
the fundamental structures appearing in the dynamics of mechanical
and control systems. In fact, the geometrical analysis of such
systems reveals what they have in common and indicates the most
suitable strategy to analyze their solutions. Finally, the
geometrical approach has provided substantial contributions to
neighboring areas, such as molecular systems, classical field
theories, control theory, engineering, etc.

Recent efforts have led to a unified framework for geometric
mechanics based on a new structure, namely a \textit{Lie algebroid}
(see Section \ref{section1}), which represents the phase space for
lagrangian mechanics and whose dual is the phase space for
hamiltonian mechanics. These ideas were introduced in a pioneering
paper by A. Weinstein \cite{Weinstein99} (see also \cite{TD69})
where the equations of motion were derived from a Lagrangian
function given on a Lie algebroid. This was done using the linear
Poisson structure on the dual of the Lie algebroid and the Legendre
transformation associated with that regular Lagrangian. The unifying
feature of the Lie-algebroid formalism is particularly relevant for
the class of Lagrangian systems invariant under the action of a Lie
group of symmetries (see \cite{LMM} for a survey on the subject; see
also \cite{CoLeMaMaMa,Ma}).

As it turns out, the Lie-algebroid scheme is not general enough to
include some interesting mechanical systems. On a Lie algebroid, the
Jacobi identity for the bracket of sections implies the preservation
of the associated linear Poisson bracket on its dual. However, many
interesting examples are not covered by this strong assumption, for
instance nonholonomic mechanics (see
\cite{TD12,IbLeMaMa,KooMa,TD108} and references therein). Moreover,
it would be interesting to find a general setting encompassing also
some cases of dissipation of energy (for instance, explicit
time-dependent systems, systems subjected to external forces or
mechanical systems subjected to affine nonholonomic constraints).
These reasons are our main motivation for introducing hamiltonian
mechanics on more general objects, namely skew-symmetric algebroids
equipped with a 1-cocycle; skew-symmetric algebroids will allow us
to avoid the preservation of the Poisson bracket
\cite{BLMM,GrGr,GGU0,GrLeMaMa,LeMaMa,Po2} and the 1-cocycle will
introduce a dissipative character to the dynamics (for the geometric
description of time-dependent mechanics, in terms of Lie affgebroids
or, equivalently, in terms of Lie algebroids with a 1-cocycle, see
\cite{GGU,IMMS, TD49, MaMaSo,MaSo,MaMeSa}; see also \cite{DianaT}).
Other approaches to the study of nonholonomic mechanical systems
subjected to linear constraints, in the algebroid setting, have been
also discussed in some recent papers (see
\cite{CoLeMaMar,CoMa,Me,MeLa}). In these papers, the key tool is the
notion of the prolongation of a Lie algebroid over a fibration.

Our main goal is to derive a Hamilton-Jacobi equation for the case
of skew-symmetric algebroids with a 1-cocycle. As it is well known,
Hamilton-Jacobi theory for unconstrained systems is a  useful tool
for the exact integration  of Hamilton's equations, for instance
using the technique of separation of variables (see \cite{AbMa} and
references therein). In other cases,  this theory allows us to
simplify the integration of Hamilton's equations or, at least, to
find some particular solutions of the system. To summarize the idea
for classical Hamilton's equations, consider a configuration
manifold $Q$ and a hamiltonian $H: T^*Q\to \R$. Then the
Hamilton-Jacobi equation can be written as
\[
H(q, \frac{\partial W}{\partial q})=\hbox{constant}
\]
for some function $W: Q\to \R$. If we find such a function $W$, then
the integration of  Hamilton's equations (for initial condition
along $dW(Q)$) is reduced to knowing the integral curves of a vector
field on $Q$, defined as $X_H^{dW}=T\tau_{T^*Q}\circ X_H\circ dW\in
{\mathfrak X}(Q)$, where $\tau_{T^*Q}: T^*Q\to Q$ is the canonical
projection and $X_H$ is the hamiltonian vector field associated to
$H$. Hence, from the integration of  a vector field on the
configuration space it is possible to recover some of the solutions
of the original hamiltonian system.

A similar idea is  also present in riemannian geometry when we look
for a vector field $X\in {\mathfrak X}(Q)$ verifying
$\nabla^{\mathcal G}_X X=0$ (a geodesic or  auto-parallel vector
field), where $\nabla^{\mathcal G}$ is the Levi-Civita connection
associated to a riemannian metric ${\mathcal G}$ on $Q$. Their
integral curves are geodesics, that is, solutions of the geodesic
second-order differential equations corresponding to ${\mathcal G}$
with initial conditions on $\hbox{Im\;} X\subset TQ$. Observe that,
in general, $X$ is not the gradient with respect to $\mathcal{G}$ of
a function $W\in C^{\infty}(Q)$, which would be the case if we
applied the classical Hamilton-Jacobi theorem. Hence, to recover
this situation it is necessary to generalize the classical
Hamilton-Jacobi equations.

On the other hand, recently, some of the authors of this paper
proposed a generalization of the Hamilton-Jacobi equation for
skew-symmetric algebroids (see \cite{LeMaMa}). Roughly speaking, a
skew-symmetric algebroid is a vector bundle $\tau_E : E
\longrightarrow Q$ equipped with a skew-symmetric bilinear bracket
of sections and a vector bundle morphism, $\rho_E:E\to TQ$ (the
anchor map),
 satisfying a
Leibniz-type property, i.e., a Lie algebroid structure without the
integrability property, (for more details, see Section
\ref{section1}). The existence of such a structure on $E$ is
equivalent to the existence  of a linear almost Poisson bracket on
the dual bundle $\tau_{E^*}: E^*\to Q$, or the existence of an
almost differential $d^E$ on $\tau_E : D \longrightarrow Q$ which
satisfies all the properties of an standard differential except that
$(d^E)^2$ is not, in general, zero.

Skew-symmetric algebroids were used in \cite{LeMaMa} to describe the
Hamilton-Jacobi equation of nonholonomic mechanical systems. In this
case  $E$ is determined by the linear constraints and the function
$W\in C^\infty(Q)$ is replaced by a $1$-cocycle on the dual bundle
$E^*$ (i.e., a section $\alpha$ of $E^*$ such that $d^E\alpha=0$).
With these ideas, one derives a Hamilton-Jacobi equation for
nonholonomic dynamics, illustrating the utility of this new theory
for the integration of different nonholonomic problems.
Hamilton-Jacobi theory for standard nonholonomic mechanical systems
has been also discussed in recent papers  (see \cite{CaGrMaMaMuRo,
ILM,BT}).

In this paper, we develop a Hamilton-Jacobi theory including, as
particular cases, the Hamilton-Jacobi equation for skew-symmetric
algebroids introduced in \cite{LeMaMa} and the case of auto-parallel
vector fields in riemannian geometry (Example \ref{Classical}), as
well as a great variety of new examples (time-dependent hamiltonian
systems, systems with external forces, nonholonomic mechanics with
affine constraints...). With this objective in mind, we obtain the
main result of our paper, Theorem \ref{main}, the Hamilton-Jacobi
equation for a hamiltonian system on a skew-symmetric algebroid with
a $1$-cocycle. Moreover,  our construction is preserved under the
natural morphisms of the theory. This fact is proved in Theorem
\ref{theo:morphism}. We remark that this new version of the
Hamilton-Jacobi equation is much more general than the one developed
in \cite{LeMaMa}, since here, we do not require the 1-section
solutions of the Hamilton-Jacobi equation to be closed. This fact is
extensively used in Example \ref{ex:ball}, where we find solutions
for the problem of a rolling ball in a rotating plane with
time-dependent angular velocity looking  for functions $W\in
C^\infty(Q)$ which do not satisfy $(d^E)^2 W=0$ (and thus, out of
the cases studied in \cite{LeMaMa}). Moreover, the proof of the
Theorem \ref{main} is simpler and completely independent of the one
done in \cite{LeMaMa}.

The paper is structured as follows. In Section \ref{section1}, we
discuss some aspects of the geometry of skew-symmetric algebroids in
the presence of a 1-cocycle. Moreover, given a hamiltonian section
$h$ of the AV-bundle associated with the skew-symmetric algebroid
and the 1-cocycle, we obtain Hamilton equations for $h$. In Section
\ref{section2}, we formulate and prove the Hamilton-Jacobi Theorem
for a hamiltonian system on a skew-symmetric algebroid with a
1-cocycle. In addition, we see that the Hamilton-Jacobi equation is
preserved under the natural morphisms of the theory. Finally, in the
last section, some theoretical and practical examples will
illustrate the power of these new techniques as for instance:
Hamilton-Jacobi equation for a particle on a vertical cylinder in a
uniform gravitational field with friction, for a homogeneous rolling
ball without sliding on a rotating table with time-dependent angular
velocity or for the vertical rolling disk with external forces.

\section{Skew-symmetric algebroids, 1-cocycles and Hamiltonian dynamics}\label{section1}

\subsection{Skew-symmetric algebroids and 1-cocycles}

Let $\tau_E:E\to Q$ be a vector bundle of rank $n$ over the manifold
$Q$. We denote by $\Gamma(E)$ the $C^\infty(Q)$-module of sections
of $E$. {\it A skew-symmetric algebroid structure } on $E$ is a pair
$(\lcf\cdot,\cdot\rcf,\rho)$, where
$\lcf\cdot,\cdot\rcf:\Gamma(E)\times \Gamma(E)\to \Gamma(E)$ is a $\R$-bilinear
skew-symmetric bracket on $\Gamma(E)$ and $\rho:E\to TQ$ is a vector
bundle map ({\it the anchor map}) such that
\[
\lcf \sigma,f\gamma\rcf=f\lcf \sigma,\gamma\rcf +
\rho(\sigma)(f)\gamma,\;\;\; \sigma, \gamma\in \Gamma(E), \;\;\;
f\in C^\infty(Q).
\]
Note that $\rho:E\to TQ$ induces a homomorphism of
$C^\infty(Q)$-modules which we denote also by $\rho:\Gamma(E)\to
{\frak X}(Q)$ (see \cite{BLMM,GGU0,GU1,GU2,LeMaMa,Po2}).

If the bracket $\lcf\cdot,\cdot\rcf$ satisfies the Jacobi identity
then $(E,\lcf\cdot,\cdot\rcf,\rho)$ is a {\it Lie algebroid} (see,
for instance, \cite{Mac}). In such a case, we have that the anchor
map is a morphism of Lie algebras, i.e.
\[
\rho(\lcf \sigma, \gamma\rcf)=[\rho(\sigma),\rho(\gamma)], \mbox{
for } \sigma,\gamma\in \Gamma(E).
\]

On a skew-symmetric algebroid  structure $(\lcf\cdot,\cdot\rcf,\rho)$ on the
vector bundle $\tau_E:E\to Q$ it is induced {\it the almost
differential } $d^E:\Gamma(\wedge^\bullet E^*)\to
\Gamma(\wedge^{\bullet+1}E^*)$ as a $\R$-linear map given by
\begin{equation}\label{1.0}
\begin{array}{rcl}
(d^Ef)(\sigma)&=&\rho(\sigma)(f),\\[8pt]
(d^E\alpha)(\sigma,\gamma)&=&\rho(\sigma)(\alpha(\gamma))-\rho(\gamma)(\alpha(\sigma))-\alpha(\lcf
\sigma,\gamma\rcf)
\end{array}
\end{equation}
and
\[
d^E(\beta_1\wedge\beta_2)=d^E\beta_1\wedge \beta_2 +
(-1)^k\beta_1\wedge d^E\beta_2,
\]
for $f\in C^\infty(Q)$, $\alpha\in \Gamma(E^*)$, $\sigma,\gamma\in
\Gamma(E)$, $\beta_1\in \Gamma(\wedge^kE^*)$  and $\beta_2\in
\Gamma(\wedge^\bullet E^*)$.

Note that $d^E$ is defined in a similar way that the standard differential over a manifold. However, there are important differences between them. In what follows, we will discuss some facts related with these differences.

Firstly, unlike the case of the stardard differential on a manifold,
we have that the almost differential $d^E$ of a skew-symmetric
algebroid $(E,\lcf\cdot,\cdot\rcf,\rho)$ is not, in general, a
cohomology operator, i.e., $(d^E)^2\not=0$. In fact, $(d^E)^2=0$ if
and only if $(\lcf\cdot,\cdot\rcf,\rho)$ is a Lie algebroid
structure.

For the particular case of a function $g\in C^\infty(Q)$, we deduce
from (\ref{1.0}) that $(d^E)^2g=0$ if and only if $X(g)=0$ for all
$X\in \widetilde{D}$, where $\widetilde{D}$ is the finitely
generated distribution given by
$$\widetilde{D}:= span \{ [\rho(\sigma),\rho(\gamma)] - \rho (\lcf
\sigma, \gamma \rcf)  \ : \ \sigma, \gamma \in \Gamma(E)\} \subseteq
\mathfrak{X}(Q).$$

On the other hand, if $Q$ is a connected manifold, in general, $d^Eg=0$ does not imply that $g:Q\to {\mathbb R}$ is constant. In other words, if $Q$ is a connected manifold, in general,  the vector space
$$H^0(d^E)=\{f\in C^\infty(Q)\ \mbox{such that} \ d^Ef=0\},$$  is not  isomorphic to ${\mathbb R}.$  Note that, when $E$ is a Lie algebroid, $H^0(d^E)$ is the Lie algebroid cohomology $0$-group of $E$. Even in this case, one can not guarantee that  $H^0(d^E)$ is isomorphic
${\mathbb R}$. However, if $Q$ is connected and $E$ is transitive, i.e., $\rho(E)=TQ$, then $H^0(d^E)\cong {\mathbb R}$.

In \cite{LeMaMa} the authors discuss the relation between a function
$g\in C^\infty(Q)$  being constant and $d^Eg=0$. In order to
remember these results  we introduce  the notion of completely
nonholonomic distribution (see \cite{Mont}).

\begin{definition} \ A distribution \  ${\mathcal D} \subset TQ$ \ is called {\it
completely nonholonomic} (or {\it bracket generating}) if  $\{
 X_k, [X_k,X_l],[X_i,[X_k,X_l]],... \in
\mathfrak X(Q) \ : \ X_j(q) \in {\mathcal D}_q \ \forall q \in Q\}$
spans the tangent bundle $TQ$.\end{definition}

The Lie brackets of vectors fields in ${\mathcal D}$ generate a flag
${\mathcal D} \subset {\mathcal D}^2 \subset ... \subset TQ$ with
$${\mathcal D}^2= {\mathcal D} + [{\mathcal D},{\mathcal D}], \qquad
{\mathcal D}^{r+1} = {\mathcal D}^r + [{\mathcal D}, {\mathcal D}^r]
$$where
$$[{\mathcal D},{\mathcal D}^k] = span \{[X,Y] \ : \ X \in {\mathcal
D} \mbox{ and } Y \in {\mathcal D}^k\}$$ with the spans taken over
smooth functions on $Q$ (for details, see \cite{Mont}).

Here, we have two extreme cases: on one hand, the distribution
${\mathcal D}$ can be involutive, then we have ${\mathcal
D}={\mathcal D}^2={\mathcal D}^r, \ \forall r \in \N_{>0}.$ On the
other hand, if ${\mathcal D}$ is completely nonholonomic, then there
exists $r \in \N_{>0}$ such that ${\mathcal D}^r=TQ$.

Now, consider a skew-symmetric algebroid $(E,\lcf\cdot,\cdot\rcf,\rho)$ on a manifold $Q$ and the following finitely generated distribution ${\mathcal D}$ given by
$${\mathcal D}_q=\rho_q({E_q}),\;\;\; \mbox{ for all }q\in Q. $$

\begin{proposition} \cite{LeMaMa}
If $(E, \lcf\cdot , \cdot \rcf, \rho)$ is a skew-symmetric algebroid
over a connected manifold $Q$, such that $ {\mathcal D} = \rho(E)
\subset TQ$ is a completely nonholonomic distribution, then
$H^0(d^E)\cong \R$.
\end{proposition}

However, in some examples, the distribution ${\mathcal D}$ is not
completely nonholonomic. In such a case there is  $r \in \N_{>0}$ such that
${\mathcal D}^{r-1} \subsetneq {\mathcal D}^r = {\mathcal D}^{r+1}
\subsetneq TQ.$  This distribution ${\mathcal D}^r$ is the smallest
Lie subalgebra of ${\mathfrak X}(Q)$ containing ${\mathcal D}$. Let
us consider the associated  generalized foliation over $Q.$  The leaf of this foliation over a point $q
\in Q$, is just the orbit
$$ L = \{ \phi_{t_k}^{X_k}\circ ... \circ \phi_{t_1}^{X_1} (q) \in Q
\ : \ t_i \in \R, \mbox{ and } X_i \in {\mathcal D} \mbox{ with }
i=1,...,k, \ k \in \N_{>0} \}$$ where $\phi_{t_i}^{X_i}$ is the flow
of the vector field $X_i$ at time $t_i$(see
\cite{Orbit},\cite{sussman}).

\begin{theorem} \cite{LeMaMa}
Let $(E, \lcf \cdot , \cdot \rcf, \rho)$ be a  skew-symmetric
algebroid over a manifold $Q$. Consider the leaf $L$ of ${\mathcal
D}^r$ as described above. Then \begin{enumerate} \item It is induced
a skew-symmetric algebroid structure $(\lcf \cdot , \cdot
\rcf_L,\rho_L)$ on the vector bundle $\tau_L: E_L  \rightarrow L$
with $E_L:= \cup_{q\in L} E_q$. Moreover, the distribution
$\rho_L(E_L)$ on $L$ is completely nonholonomic.
\item If $f \in C^\infty(Q )$ is a
function such that $d^Ef =0$, then its restriction to $L$ is
constant.
\end{enumerate}
\end{theorem}

Next, we will see that any skew-symmetric algebroid structure $(\lcf\cdot ,\cdot\rcf,\rho)$ on the vector bundle $\tau_E:E\to Q$  induces an almost Poisson linear bracket
$\{\cdot,\cdot\}:C^\infty(E^*)\times C^\infty(E^*)\to
C^\infty(E^*)$ on the space of functions on $E^*$, that is, $\{\cdot,\cdot\}$ is  a skew-symmetric $\R$-bilinear bracket which is a derivation in each argument with respect to the standard product of functions  and  with the extra property that
the bracket of two linear functions is again a linear function. Indeed,  this bracket is  characterized by  the following  relations (see
\cite{BLMM,GGU0,GU1,GU2,LeMaMa})
\begin{equation}\label{Poisson}
\begin{array}{c}
\{\widehat{{\sigma}},\widehat{{\gamma}}\}=-\widehat{\lcf
{\sigma},{\gamma}\rcf }, \qquad \{f\circ \tau_{{E}^*},
\widehat{\sigma}\}=\rho({\sigma})(f)\circ \tau_{{E}^*}\\[8pt]
 \{f\circ \tau_{{E}^*},g\circ \tau_{{E}^*}\}=0,
 \end{array}
 \end{equation}
for all $\sigma,\gamma\in \Gamma(E)$ and $f,g\in C^\infty(Q)$ and
where $\widehat{\zeta}:E^*\to \mathbb{R}$ is the linear
function associated with the section $\zeta\in \Gamma(E).$

Now, we will endow our skew-symmetric algebroid with an additional
structure: a distinguished section $\phi\in \Gamma(E^*)$  which
allows us to consider some interesting examples.

Let us consider a section  $\phi$ of $E^*$.  Denote by $\phi^\vee\in {\mathfrak X}(E^*)$ the vertical
lift of the section $\phi\in \Gamma(E^*)$, that is, the vector field
defined by
\[
\phi^\vee(\alpha)=(\phi(\tau_{E^*}(\alpha)))_{\alpha}^\vee,\;\;\;
\forall \alpha\in E^*,\] where
$\;_{\alpha}^\vee:E^*_{\tau_{E^*}(\alpha)} \to
T_{\alpha}(E^*_{\tau_{E^*}(\alpha)})$ is the canonical isomorphism
between the vector spaces $E^*_{\tau_{E^*}(\alpha)}$ and
$T_{\alpha}(E^*_{\tau_{E^*}(\alpha)}).$ Note that
\begin{equation}\label{vertical-lineal}
\phi^\vee(\widehat{\sigma})=\phi(\sigma)\circ \tau_{E^*}=\widehat{\sigma}\circ \phi\circ\tau_{E^*},
\end{equation}
for all $\sigma\in \Gamma(E)$.

Using (\ref{1.0}), (\ref{Poisson}) and (\ref{vertical-lineal}), we obtain the following formula which describes the differential of $\phi$   in terms of the linear bracket $\{\cdot,\cdot\}$ on $E^*$
\begin{equation}\label{d-b}
\begin{array}{rcl}
d^E\phi(\gamma,\sigma)\circ \tau_{E^*}&=&-(\{\phi^\vee(\widehat{\gamma}),\widehat{\sigma}\}
+ \{\widehat{\gamma},\phi^\vee(\widehat{\sigma})\}-\phi^\vee(\{\widehat{\gamma},\widehat{\sigma}\})\\[8pt]&=&-(\{\widehat{\gamma},\widehat{\sigma}\circ \phi\circ \tau_{E^*}\} -\{\widehat{\gamma},\widehat{\sigma}\}+\{\widehat{\gamma}\circ \phi\circ \tau_{E^*},\widehat{\sigma}\}),
\end{array}
\end{equation}
for all $\gamma, \sigma\in \Gamma(E)$. Thus,  $\phi\in \Gamma(E^*)$ is a $1$-cocycle, i.e., $d^E\phi=0$, if and only if \begin{equation}\label{invarian}
\phi^\vee(\{\varphi_1,\varphi_2\})=\{\phi^\vee(\varphi_1),\varphi_2\}
+ \{\varphi_1,\phi^\vee(\varphi_2)\},
\end{equation}
for all $\varphi_1,\varphi_2\in C^\infty(E^*)$. Moreover, equation
(\ref{invarian}) is equivalent to the fact that the linear bivector
$\Pi_{E^*}$ on $E^*$ associated with the bracket $\{\cdot,\cdot\}$
is invariant with respect to the vector field $\phi^\vee\in
{\mathfrak X}(E^*)$, i.e.,
\begin{equation}\label{invariante}
{\mathcal L}_{\phi^\vee}\Pi_{E^*}=0.
\end{equation}
In fact, from (\ref{Poisson}), (\ref{invarian}) and
(\ref{invariante}) one may conclude the following result:

\begin{proposition}
Let $\tau_E:E\to Q$ be a vector bundle over $Q$ and $\phi\in
\Gamma(E^*)$ a section of the dual bundle of $E$. Then, the
following statements are equivalent:
\begin{enumerate}
\item $E$ admits a skew-symmetric algebroid structure $(\lcf\cdot,\cdot\rcf,\rho)$  such
that $d^E\phi=0$.
\item There is a linear bivector $\Pi_{E^*}$ on $E^*$ which is invariant
with respect to the vertical lift $\phi^\vee\in {\mathfrak
X}(E^*)$ of $\phi$.
\end{enumerate}
\end{proposition}

\subsection{Hamiltonian dynamics}\label{Sec:HamDyn}

Let $(E,\lcf\cdot,\cdot\rcf,\rho)$ be a skew-symmetric algebroid
over $Q$ of rank $n$, and  $\phi\in \Gamma(E^*)$ be a $1$-cocycle of
$E$. Denote by $\widehat{\phi}:E\to \mathbb{R}$ the corresponding
linear function induced by $\phi$ on $E$ and suppose that, for all
$q\in Q,$ $\widehat{\phi}_{|E_q}\not=0.$ Then, one may consider the
affine bundle
$$\tau_{\mathcal A}:{\mathcal A}:=\widehat{\phi}^{-1}(1)\to Q$$ of rank $n-1$ with associated
vector bundle $\tau_V:V:=\widehat{\phi}^{-1}(0)\to Q.$ Note that $V$
is a skew-symmetric algebroid over $Q$ with structure
$(\lcf\cdot,\cdot\rcf_V,\rho_V)$ given by
$$i_V\circ \lcf \sigma,\gamma\rcf_V=\lcf i_V(\sigma),i_V(\gamma)\rcf,\;\;\;
\rho_V(\sigma)=\rho(i_V(\sigma))$$
for all $\sigma,\gamma\in \Gamma(V)$, where $i_V:V\to E$ is the
canonical inclusion.  Thus, we have the
corresponding linear almost Poisson 2-vector $\Pi_{V^*}$ on $V^*$.

On the other hand, the map $\mu:= i_V^* : E^*\to V^*$ defines an
affine bundle of rank $1$ modeled over the trivial vector bundle
$pr_1:V^*\times \mathbb{R}\to V^*$ ({\it an AV-bundle } in the
terminology of \cite{GGU}). Using (\ref{Poisson}) and the fact that
the canonical inclusion $i_V:V\to E$ is a skew-symmetric algebroid
monomorphism, we deduce that $\mu:E^*\to V^*$ is an almost Poisson
morphism. Thus, if $\alpha_q\in E^*_q$, we have that the following
diagram is commutative
$$
\xymatrix{
T^*_{\alpha_q}E^*\ar[rr]^{\#_{\Pi_{E^*}}(\alpha_q)}&&T_{\alpha_q}{E}^*\ar[d]^{T_{\alpha_q}\mu}\\
T^*_{\mu(\alpha_q)}V^*\ar[u]^{T^*_{\alpha_q}\mu}\ar[rr]^{\#_{\Pi_{V^*}}(\mu(\alpha_q))}&&T_{\mu(\alpha_q)}V^*}
$$

Here, $\#_{\Pi_{E^*}}:T^*E^*\to T E^*$ (respectively,
$\#_{\Pi_{V^*}}:T^*V^*\to TV^*)$ is the morphism of
$C^\infty(Q)$-modules induced by the almost Poisson bivector
$\Pi_{E^*}$ (respectively, $\Pi_{V^*}$).

Using again (\ref{Poisson}), we also deduce that
$$
d^Vf=\mu\circ d^Ef \mbox{ and } d^V(\mu\circ
\alpha)=\wedge^2\mu\circ d^E\alpha, \;\;\; f\in C^\infty(M),\;\;\;
\alpha\in \Gamma(E^*), $$ where $\wedge^2 \mu : \wedge^2E^*
\rightarrow \wedge^2V^*$ is the extension of $\mu$ to the
corresponding vector bundles.

Note that the set of the global sections $\Gamma(\mu)$, of $\mu$, is
an affine space modeled over $C^\infty(V^*)$. In addition, if $h\in
\Gamma(\mu)$ then $\mu(\alpha_q - h(\mu(\alpha_q)))=0,$ for $q \in Q$
and $\alpha_q \in E^*_q$. Thus, one may define a function $F_h\in
C^\infty(E^*)$ characterized by
\begin{equation}\label{h}
\alpha_q-h(\mu(\alpha_q))=F_h(\alpha_q)\phi(q),
\end{equation}
 for all $q\in Q$ and for all $\alpha_q\in E^*_q$. Moreover, we have
$$(\phi^\vee(F_h))(\alpha_q) = \left(
 \frac{d}{dt}_{|t=0} F_h (\alpha_q+t\phi(q)) \right) =
 1.$$ Therefore, it follows that
 $\phi^\vee(F_h)=1$. In fact,  there is a one-to-one correspondence between $\Gamma(\mu)$
and the set of functions $F$ on $E^*$ which satisfy the relation
\begin{equation}\label{F}
\phi^\vee(F)=1,
\end{equation}
(see \cite{GGU}).

In what follows, we will associate to each section $h \in
\Gamma(\mu)$, a vector field on $V^*$. From (\ref{invarian}) and
(\ref{F}), we deduce that, for every section $h$ of the bundle
$\mu:E^*\to V^*$ and $G\in C^\infty(V^*),$  the function $\{G\circ
\mu,F_h\}$ is $\mu$-projectable (note that $Ker T_{\alpha_q}\mu =
<\phi^\vee(\alpha_q)>$). Thus, for each $h \in \Gamma(\mu)$ we can
consider a vector field $R_h$ on $V^*$ which is characterized by
\begin{equation}\label{2.11'}
R_h(G)\circ \mu=\{G\circ \mu, F_h\},\;\;\;\; G\in C^\infty(V^*).
\end{equation}
This vector field is called the {\it hamiltonian vector field
associated with the section $h$}. If ${\mathcal
H}^{\Pi_{E^*}}_{F_h}\in {\mathfrak X}(E^*)$ is the hamiltonian
vector field associated with the function $F_h\in C^\infty(E^*)$
with respect to the linear almost Poisson bracket $\Pi_{E^*}$,
i.e.,
\[
{\mathcal H}^{\Pi_{E^*}}_{F_h}=-i_{dF_h}\Pi_{E^*}
\]
then, from (\ref{2.11'}), we deduce that
\begin{equation}\label{Hamilton}
R_h\circ \mu=T\mu\circ {\mathcal H}_{F_h}^{\Pi_{E^*}}.
\end{equation}

The integral curves of the vector field $R_h$ are the solutions of
{\it the Hamilton equations } for $h$.
\subsection{Local expressions}\label{section-2.3}

Let $E$ be a vector bundle on $Q$ of rank $n$, with a skew-symmetric
algebroid structure $(\lcf\cdot,\cdot\rcf,\rho)$.

Fixed a section  $\phi$  of  $E^*$ such that $d^E \phi =0$ and
$\widehat{\phi}_{|E_q}\not=0$ for all $q\in Q$. Then it is induced a
local basis $\{e_0, e_a\}_{a=1,\dots, n-1}$ of $E$ adapted to the
$1$- section $\phi$ in the sense that $\phi(e_0)=1$ and
$\phi(e_a)=0.$  In terms of this basis we have the {\it local
structure functions}, $\C_{ab}^c,\rho_a^i, \C_{0b}^c,\rho_0^i \in
C^\infty(Q)$ of $E$ defined by
$$
\lcf e_a,e_b\rcf={\mathcal C}_{ab}^c e_c, \quad \lcf
e_0,e_b\rcf={\mathcal C}_{0b}^c e_c \quad \mbox{and} \quad
\rho(e_a)=\rho_a^i\frac{\partial }{\partial q^i}, \quad
\rho(e_0)=\rho_0^i\frac{\partial }{\partial q^i} .$$

Note that the condition $d^E\phi=0$ implies that $\C_{ab}^0=0$, for all
$a,b \in \{0,\dots, n-1\}$.

Moreover, with respect to the induced local coordinates
$(q^i,p_0,p_a)$ on $E^*$, the local expressions of the vector field
$\phi^\vee\in {\mathfrak X}(E^*)$ and the linear almost Poisson
bivector $\Pi_{E^*}$ are
$$\begin{array}{rcl}\phi^\vee&=&\displaystyle\frac{\partial }{\partial p_0},\\
\Pi_{E^*}&=&\rho_0^i\displaystyle\frac{\partial }{\partial
q^i}\wedge \displaystyle\frac{\partial }{\partial p_0} +
\rho_a^i\displaystyle\frac{\partial }{\partial q^i}\wedge
\displaystyle\frac{\partial }{\partial p_a}-\C_{0b}^c
p_c\displaystyle\frac{\partial }{\partial
p_0}\wedge\displaystyle\frac{\partial }{\partial p_b}
-\displaystyle\frac{1}{2}\C_{ab}^c p_c \displaystyle\frac{\partial
}{\partial p_a}\wedge \displaystyle\frac{\partial }{\partial
p_b}.\end{array}$$

If $(q^i,p_a)$ are the corresponding coordinates of
$V^*,$ the local expression of $\mu:E^*\to V^*$ is
$$\mu(q^i,p_0,p_a)=(q^i,p_a).$$

Let $h:V^*\to E^*$ be a section of $\mu$ whose local expression is
$$h(q^i,p_a)= (q^i,-H(q^i,p_b),p_a)$$ where  $H$ is a local function of $V^*.$ Then the associated function $F_h :E^* \rightarrow \R$ is
\begin{equation}\label{1.11'}
F_h(q^i,p_0,p_a)=p_0+H(q^i,p_a).
\end{equation}

Moreover, the local expression of the Hamiltonian vector field
associated with this section $h:V^*\to E^*$ is given by
$$
R_{h}=(\rho_0^i + \rho_a^i \frac{\partial H}{\partial
p_a})\frac{\partial }{\partial q^i} + (-\rho_b^i \frac{\partial
H}{\partial q^i} + (\C_{0b}^c + \C_{ab}^c \frac{\partial H
}{\partial p_a})p_c)\frac{\partial }{\partial p_b}.
$$

Thus, the Hamilton equations are
\[
\frac{dq^i}{dt}=\rho_0^i + \rho_a^i\frac{\partial H}{\partial p_a},
\qquad  \frac{dp_b}{dt}=-\rho_b^i\frac{\partial H}{\partial q^i} +
(\C_{0b}^c + \C_{ab}^c\frac{\partial H}{\partial p_a})p_c.
\]
A Lagrangian version of these equations was considered in \cite{IMMS} (see also \cite{DianaT}).

It is important to note that the previous dynamics on $V^*$ has a
{\it dissipative character}. In fact, in the case when the AV-bundle
$\mu:E^*\to V^*$ is trivial, then the local function $H$ is global
and it is the hamiltonian function on $V^*$. In addition,
\[
R_h(H)\circ \mu=\{H\circ \mu,F_h\}\not=0.
\]
The local expression of this dissipative term  is
\[
\{H\circ \mu,F_h\}=\rho_0^i\frac{\partial H}{\partial
q^i}+\C_{0b}^cp_c\frac{\partial H}{\partial p_b}.
\]

\subsection{Examples}
Next, we will describe two interesting examples of skew-symmetric algebroids which will be useful for the mathematical description of the mechanical systems considered in this paper.


\begin{example}\label{ejemplo} {\rm Consider a  Lie algebroid structure $(\lcf\cdot,\cdot\rcf,\rho)$
(or more generally a skew-symmetric algebroid structure)  on a
vector bundle $\tau_{\bar{E}}: {\bar{E}}\to Q$ of rank $n-1$ and $F:{\bar{E}}\to
{\bar{E}}$ a homomorphism of vector bundles (over the identity of
$Q$). Then, on the vector bundle $\tau_{\mathbb{R}\times
{\bar{E}}}:\mathbb{R}\times \bar{E}\to Q$, it is induced a
skew-symmetric algebroid structure $(E:= \mathbb{R}\times \bar{E},
\lcf\cdot,\cdot\rcf_{\mathbb{R}\times
\bar{E}},\rho_{\mathbb{R}\times \bar{E}})$ given by
\begin{equation}\label{1Ex}\begin{array}{rcl} \lcf (f,\sigma),(g,\gamma)\rcf_{\mathbb{R}\times \bar{E}}&=&(\rho(\sigma)(g)-\rho(\gamma)(f),
\lcf \sigma,\gamma\rcf+gF(\sigma)-fF(\gamma)),
\\[5pt] \rho_{\mathbb{R}\times \bar{E}}(f,\sigma)&=&\rho(\sigma),
\end{array} \end{equation}
for all $(f,\sigma),(g,\gamma)\in \Gamma(E)= C^\infty(Q)\times
\Gamma(\bar{E})\cong \Gamma(\mathbb{R}\times \bar{E})$.

Note that the space $\Gamma(\wedge^2(\mathbb{R}\times \bar{E}^*))$
may be identified with $ \Gamma(\bar{E}^*)\oplus\Gamma(\wedge^2
\bar{E}^*)$. Under this identification, for $(f, \alpha)\in
C^{\infty}(Q)\times\Gamma(\bar{E}^*)\cong \Gamma(\mathbb{R}\times
\bar{E}^*)$, we obtain that
\begin{equation}\label{dER}
\begin{array}{rcl}
d^{\mathbb{R}\times \bar{E}}(f,\alpha)&=&(F^*\alpha - d^{\bar{E}}f,
d^{\bar{E}}\alpha),\end{array}
\end{equation}
where $F^*\alpha$ is a section of $\bar{E}^*$ defined by
$$
(F^*\alpha)(\sigma)=\alpha(F(\sigma)),\,\;\; \forall\sigma\in
\Gamma(\bar{E}).$$

From (\ref{dER}), we have that  $(1,0) \in \Gamma(\R\times \bar
E^*)$ is a $1$-cocycle, and  its vertical lift is just the vector
field $\frac{\partial}{\partial p_0}$ on
$\mathbb{R}\times\bar{E}^*$, with $p_0$ the global coordinate on
$\mathbb{R}$.

In this case, the linear almost Poisson bivector on $\mathbb{R}\times \bar{E}^*$ is given by
\[
\Pi_{\mathbb{R}\times\bar{E}^*}=\Pi_{\bar{E}^*} + \frac{\partial }{\partial
p_0}\wedge (F^*)^\vee,
\]
where $\Pi_{\bar{E}^*}$ is the linear Poisson bivector on
$\bar{E}^*$ induced by the Lie algebroid structure on $\bar{E}$ and
$(F^*)^\vee$ is the vector field on $\bar{E}^*$ defined by
\[
(F^*)^\vee:\bar{E}^*\to T\bar{E}^*,\;\;\;\
(F^*)^\vee(\alpha)=F^*(\alpha)_\alpha^\vee.
\]

Here,  $F^*:\bar{E}^*\to \bar{E}^*$ is the dual morphism of $F:\bar{E}\to \bar{E}.$

 On the other hand, the affine bundle associated with
$(\mathbb{R}\times \bar{E},\lcf\cdot,\cdot\rcf_{\mathbb{R}\times
\bar{E}}, \rho_{\mathbb{R}\times \bar{E}},(1,0))$ is just
$\tau_{\{1\}\times \bar{E}}:\{1\}\times\bar{E} \cong \bar E \to Q$
with associated vector bundle $\tau_{\{0\}\times
\bar{E}}:\{0\}\times\bar{E} \cong \bar E\to Q$.

Moreover, the associated AV-bundle $\mu:\mathbb{R}\times\bar{E}^*\to
\bar{E}^*$ is just the trivial affine bundle over $\bar{E}^*$ of
rank $1$. Therefore, there is a one-to-one correspondence between sections
 $h:\bar{E}^*\to
\mathbb{R}\times\bar{E}^*$ of this affine bundle and  functions
$H:\bar{E}^*\to \mathbb{R}.$ Furthermore, the hamiltonian vector field
associated with $h$ is just
\begin{equation}\label{1.14} R_h={\mathcal
H}_H^{\Pi_{\bar{E}^*}} - (F^*)^\vee,
\end{equation}
 where ${\mathcal
H}_H^{\Pi_{\bar{E}^*}}\in {\mathfrak X}(\bar{E}^*)$ is the hamiltonian vector
field associated with $H$ with respect to the linear Poisson
structure $\Pi_{\bar{E}^*}.$

Note that, in this case, the dissipative term is given by
\[
R_h(H)=-(F^*)^\vee(H).
\]

Now, we will give some expressions in coordinates of the above
objects. Let us consider local coordinates $(q^i)$
 on the manifold $Q$ and a local basis $\{\bar{e}_a\}_{a=1,\dots, n-1}$ of $\Gamma(\bar{E}).$

A local basis of sections of ${E}=\R\times\bar{E}$ is $\{e_0,
e_a\}_{a=1,\ldots, n-1}$, where $e_0=(1,0)$ and $e_a=(0,\bar{e}_a)$.
Locally, the homomorphism of vector bundles $F: \bar{E}\to \bar{E}$
is given by the functions $F^b_a\in C^{\infty}(Q)$, where
 $F(e_a)=F^b_a e_b$.
 Then,
$$ \begin{array}{rclcrcl}
\lcf e_0, e_a\rcf_{\R\times \bar{E}}&=& -F_a^b e_b, & \ & \lcf e_a,
e_b\rcf_{\R\times \bar{E}}&=& {\mathcal C}_{ab}^c e_c\\ [8pt]
\rho_E(e_0)&=&0, & \ &
\rho_E(e_a)&=&\rho_a^i\frac{\partial}{\partial q^i}.
\end{array}$$
The linear almost Poisson bivector on $E^*=\R\times \bar{E}^*$ is now:
\begin{eqnarray*}
\Pi_{E^*}&=& \rho_a^i\displaystyle\frac{\partial }{\partial
q^i}\wedge \displaystyle\frac{\partial }{\partial p_a}+F_b^c
p_c\displaystyle\frac{\partial }{\partial
p_0}\wedge\displaystyle\frac{\partial }{\partial p_b}
-\displaystyle\frac{1}{2}\C_{ab}^c p_c \displaystyle\frac{\partial
}{\partial p_a}\wedge \displaystyle\frac{\partial }{\partial p_b}.
\end{eqnarray*}
Given a hamiltonian function $H: \bar{E}^*\to \R$ then the Hamilton
equations of motion are:
\begin{eqnarray*}
\frac{dq^i}{dt}&=& \rho^i_a\frac{\partial H}{\partial  p_a}\\
\frac{dp_b}{dt}&=&-\rho^i_b\frac{\partial H}{\partial
q^i}+\C^c_{ab}p_c\frac{\partial H}{\partial p_a}-F_b^a  p_a
\end{eqnarray*}

}
\end{example}


\begin{example}\label{projector}{\rm  (see \cite{LeMaMa}) Let $(E,\lcf\cdot,\cdot\rcf,\rho)$ be a Lie algebroid on a manifold $Q$.
Suppose that we have a vector subbundle $D$ of $E$ and a projector,
i.e., a vector bundle morphism ${\mathcal P}:E\to D$ (over the
identity of $Q$) such that ${\mathcal P}_{|D}=id.$  Denote by
$i_D:D\to E$ the natural inclusion from $D$ to $E.$ Then, we may
induce a skew-symmetric algebroid structure on $D$ as follows
$$\lcf \sigma,\gamma\rcf_D={\mathcal P}(\lcf i_D\circ \sigma,i_D\circ \gamma\rcf),\;\;\;\rho_D(\sigma)=\rho(i_D\circ \sigma),$$
for all $\sigma,\gamma\in \Gamma(D)$.
Note that, in general, $(\lcf\cdot,\cdot\rcf_D,\rho_D)$ is not a Lie algebroid structure on $D$.} \end{example}


\section{Hamilton-Jacobi equation, skew-symmetric algebroids with a 1-cocycle
and linear almost Poisson morphisms}\label{section2}

\subsection{Hamilton-Jacobi equation}
Let $\tau_E:E\to Q$ be a vector  bundle, of rank $n$, on the
manifold $Q$ with a skew-symmetric algebroid structure
$(\lcf\cdot,\cdot\rcf,\rho).$

Consider $\phi\in \Gamma(E^*)$  a $1$-cocycle of $E$ which satisfies  the following condition:
$$
\widehat{\phi}_{|E_q}\not=0, \mbox{ for  all  }q\in Q. $$ Then, as
we have shown in Section \ref{Sec:HamDyn}, the vector bundle
$\tau_V:V:=\widehat{\phi}^{-1}(0)\to Q$ of rank $n-1$ admits a
skew-symmetric algebroid structure which we denote by
$(\lcf\cdot,\cdot\rcf_V,\rho_V).$

If $h$ is a section of the corresponding AV-bundle $\mu:E^*\to
V^*$ then $(E,\lcf\cdot,\cdot\rcf,\rho,\phi,h)$ is said to be a
{\it Hamiltonian system on $E$.}

In such a case, for each section $\alpha$ of $V^*$, one may define a vector field $R_h^\alpha$ on $Q$ as follows
\begin{equation}\label{Rha}
R_h^\alpha=T\tau_{V^*}\circ R_h\circ \alpha,
\end{equation}
where $R_h$ is the hamiltonian vector field associated with the
section $h$.

On the other hand, we may introduce the following map
$$\widetilde{\;\;\;\;}:\Omega^1(E^*)\to \Gamma(E),\;\;\; \omega\to \widetilde{\omega},$$
where $\widetilde{\omega}$ is characterized by
\begin{equation}\label{3.2'}
\beta(\widetilde{\omega})=\omega(\beta^\vee)\circ  h\circ \alpha,
\qquad \mbox{for all} \ \beta\in\Gamma(E^*).
\end{equation}
We remark that $\widetilde{\omega}\in \Gamma(E)$ since, if
 $f\in C^\infty(Q)$ then $(f\beta)^\vee=(f\circ \tau_{E^*})\beta^\vee$ and
\begin{equation}\label{3.2''}
\tau_{E^*}\circ h \circ \alpha=id.
\end{equation}
Moreover, it follows that
\begin{equation}\label{3.2'''}
\widetilde{d\widehat{\gamma}}=\gamma,\;\;\; \widetilde{d(f\circ
\tau_{E^*})}=0,\;\;\; \mbox{ and } \; \; \; \widetilde{F\omega} =
(F\circ h\circ \alpha) \widetilde \omega,
\end{equation} for $\gamma\in \Gamma(E)$, $f\in C^\infty(Q)$, $F\in C^\infty(E^*)$ and $\omega \in \Omega^1(E^*)$. We will denote by $\zeta_h^\alpha$ the
section of $E$ given by
\begin{equation}\label{3.2}
\zeta_h^\alpha=\widetilde{dF_h}.
\end{equation}

Now, we state the main result of this paper.

\begin{theorem}\label{main}
Let $(E,\lcf\cdot,\cdot\rcf,\rho,\phi,h)$ be a {Hamiltonian system
} on $E$. If $\alpha\in \Gamma(V^*)$, we have that the following
statements are equivalent:
\begin{enumerate}
\item[$(i)$] If $c:I\to Q$ is an integral curve of $R^\alpha_h\in
{\mathfrak X}(Q)$ then $ \alpha\circ c:I\to V^*$ is an integral
curve of $R_h\in {\mathfrak X}(V^*).$
\item[$(ii)$] $\alpha\in \Gamma(V^*)$ satisfies the {\bf
Hamilton-Jacobi equation}, i.e.,
$$\mu\circ ( i_{\zeta_{h}^{\alpha}}d^E(h\circ \alpha))=0.$$
 \end{enumerate}
\end{theorem}

\begin{proof}
From (\ref{Poisson}), (\ref{d-b}) and (\ref{3.2''}), we deduce that
for all $\gamma,\sigma\in \Gamma(E)$
\begin{equation}
\begin{array}{rcl}\label{3.3'}
d^E(h\circ \alpha)(\gamma,\sigma)&=&(\rho(\gamma)(\widehat{\sigma}\circ h\circ \alpha)\circ \tau_{E^*} + \{\widehat{\gamma},\widehat{\sigma}\}-\rho(\sigma)(\widehat{\gamma}\circ h\circ \alpha)\circ \tau_{E^*})\circ h\circ \alpha\\[5pt]
&\kern-60pt=& \kern-40pt-\Pi_{E^*}(d\widehat{\gamma},(h\circ \alpha\circ\tau_{E^*})^*(d\widehat{\sigma}))\circ h\circ\alpha + \Pi_{E^*}(d\widehat{\gamma}-(h\circ \alpha\circ \tau_{E^*})^*(d\widehat{\gamma}),d\widehat{\sigma})\circ h\circ \alpha.
\end{array}
\end{equation}

Moreover, using (\ref{3.2'''}), it follows that
\begin{equation}\label{3.3''}
0=d^E(h\circ \alpha)(\widetilde{d(f\circ
\tau_{E^*})},\widetilde{dF})\end{equation} and by (\ref{Poisson}),
(\ref{3.2''}) we obtain also that \begin{eqnarray}
0&=&-\Pi_{E^*}(d(f\circ \tau_{E^*}),(h\circ
\alpha\circ\tau_{E^*})^*(dF))\circ h\circ \alpha\; \label{3.31} \\ &
& \ + \Pi_{E^*}(d(f\circ \tau_{E^*}) -(h\circ \alpha\circ
\tau_{E^*})^*(d(f\circ \tau_{E^*})), dF)\circ h\circ \alpha
\nonumber
\end{eqnarray}
for all $F\in C^\infty(E^*)$ and $f\in C^\infty(Q).$

Therefore, from (\ref{3.2''}), (\ref{3.2'''}), (\ref{3.3'}),
(\ref{3.3''}) and (\ref{3.31}), we conclude that for all
$\omega_1,\omega_2\in \Omega^1(E^*)$
\begin{equation}\label{w}
\begin{array}{rcl}
d^E(h\circ \alpha)(\widetilde{\omega_1},\widetilde{\omega_2})&=&-\Pi_{E^*}(\omega_1,(h\circ \alpha\circ \tau_{E^*})^*\omega_2)\circ h\circ \alpha\\&& +
\Pi_{E^*}(\omega_1-(h\circ \alpha\circ \tau_{E^*})^*\omega_1,\omega_2)\circ h\circ \alpha.
\end{array}
\end{equation}

Next, we will denote by $i_V:V\to E$ the natural inclusion.
Then, if $\sigma\in \Gamma(V)$, considering in (\ref{w})  the particular case when $\omega_1=dF_h$ and $\omega_2=d(\widehat{\sigma}\circ \mu)=d(\widehat{i_V\circ \sigma})$, we have that
$$
d^E(h\circ \alpha)(\zeta_{h}^{\alpha},i_V\circ \sigma)={\mathcal H}_{F_h}^{\Pi_{E^*}}(\widehat{\sigma}\circ \alpha\circ \tau_{V^*}\circ \mu)\circ h\circ \alpha -
{\mathcal H}^{\Pi_{E^*}}_{F_n}(\widehat{\sigma}\circ \mu)\circ h\circ \alpha.$$
Note that $(h\circ \alpha\circ \tau_{E^*})^*(dF_h)=0$ (see (\ref{h})).

Now, using  (\ref{Hamilton}) we conclude that
\begin{equation}\label{3.4'}
d^E(h\circ \alpha)(\zeta_{h}^{\alpha},i_V\circ \sigma)= (T\alpha \circ R_h^\alpha)(\widehat{\sigma})-R_h(\widehat{\sigma})\circ \alpha.\end{equation}

Next, we remark that statement $(i)$ in the theorem is equivalent to
\begin{equation}\label{relation}
T\alpha\circ R_h^\alpha=R_h\circ \alpha.
\end{equation}

So, if this relation holds then, from (\ref{3.4'}), we deduce that $\mu\circ (i_{\zeta_h^\alpha}d^E(h\circ \alpha))=0.$

Conversely,  suppose that
\begin{equation}\label{3.4'''}
d^E(h\circ \alpha)(\zeta_{h}^{\alpha},i_V\circ \sigma)=0,\;\,\;\; \mbox{ for all }\sigma\in \Gamma(V).
\end{equation}
In order to prove (\ref{relation}), it is sufficient to see that the following two relations are satisfied:
\begin{enumerate}
\item[$(a)$] $R_h^\alpha(\widehat{\sigma}\circ \alpha)=R_h(\widehat{\sigma})\circ \alpha,$ for all $\sigma\in \Gamma(V),$
\vspace{6pt}
\item [$(b)$] $R_h^\alpha(f)=R_h(f\circ \tau_{V^*})\circ\alpha,$ for all $f\in C^\infty(Q)$.
\end{enumerate}
Statement $(a)$ is a consequence  of  Eq. (\ref{Rha}) and $(b)$
follows from (\ref{3.4'}) and (\ref{3.4'''}).
\end{proof}

In what follows, we write the local expression of the
Hamilton-Jacobi equation.

Consider local coordinates $(q^i)$ on $Q$ and a local basis
$\{e_0,e_a\}_{a=1,\dots, n-1}$ of $\Gamma(E)$ adapted to the
$1$-cocycle $\phi$ as in Subsection \ref{section-2.3}

Denote by $(q^i,p_0,p_a)$ (respectively, $(q^i,p_a)$) the
corresponding local coordinates on $E^*$ (respectively, $V^*$).
Then, the section $\alpha:Q\to V^*$ and the hamiltonian section
$h:V^*\to E^*$  are written in terms of these coordinates as
\[
\alpha(q^i)=(q^i,\alpha_a(q^i)),\,\;\;
h(q^i,p_a)=(q^i,-H(q^i,p_a),p_a). \]

On the other hand, a 1-form $\omega \in \Omega^1(E^*)$ can be
written in these coordinates as $\omega=\omega^i dq^i + \omega^a
dp^a + \omega^0 dp_0$ with $\omega^i, \omega^a, \omega^0 \in
C^\infty(E^*)$. Therefore, from (\ref{3.2'}) we obtain that the
section $\widetilde \omega \in \Gamma(E)$ is given by
$$\widetilde{\omega}=(\omega^a\circ h\circ \alpha) \, e_a+ (\omega^0\circ
h\circ \alpha)\, e_0.$$ Thus,
\begin{equation}\label{local-section}
\zeta_{h}^{\alpha}= \widetilde{d F_h} = e_0 + (\frac{\partial
H}{\partial p_a}\circ \alpha) e_a.
\end{equation}
Now, from  (\ref{1.0}) and (\ref{local-section}), the  Hamilton
Jacobi equation is given locally as follows
$$(\rho_0^i + \rho_b^i\frac{\partial H}{\partial p_b})\frac{\partial \alpha_a}{\partial q^i} + (\rho_a^i\frac{\partial H}{\partial q^i}-(C_{0a}^c-
C_{ab}^c\frac{\partial H}{\partial p_b})\alpha_c)=0,
$$
for all $a=1,\dots, n-1.$

\medskip

To finish this subsection, we will show some consequences of Theorem
\ref{main} which will be useful for the next examples.
\begin{corollary}\label{c1}
Let $(E,\lcf\cdot,\cdot\rcf,\rho,\phi,h)$ be a {Hamiltonian system }
on $E$. If $\beta\in \Gamma(E^*)$, then the following statements
are equivalent:
\begin{enumerate}
\item[$(i)$] If $c:I\to Q$ is an integral curve of $R^{\mu\circ \beta}_h\in
{\mathfrak X}(Q)$ then $ \mu\circ\beta\circ c:I\to V^*$ is an integral
curve of $R_h\in {\mathfrak X}(V^*).$
\item[$(ii)$] $\beta\in \Gamma(E^*)$ satisfies the {\bf
Hamilton-Jacobi equation}, i.e.,
$$\mu\circ i_{\zeta_{h}^{\mu\circ \beta}}d^E\beta + d^V(F_h\circ \beta)=0.$$
 \end{enumerate}
\end{corollary}
\begin{proof}
Using (\ref{h}), we deduce that the Hamilton-Jacobi equation for $\mu\circ \beta\in \Gamma(V^*)$ is
\begin{equation}\label{equ-1}\mu\circ (i_{\zeta_{h}^{\mu\circ \beta}}d^E(\beta - (F_h\circ \beta)\phi))=0.\end{equation}
Since $\phi$ is a $1$-cocycle then
\begin{equation}\label{equ-2}d^E(\beta-(F_h\circ \beta)\phi)=d^E\beta -d^E(F_h\circ \beta)\wedge \phi.\end{equation}
On the other hand, from (\ref{3.2'}), it follows that
$$
\phi(\zeta_{h}^{\mu\circ \beta})=1$$ and therefore, using
(\ref{equ-2}), we obtain that (\ref{equ-1}) is equivalent to
$$\mu\circ i_{\zeta_{h}^{\mu\circ \beta}}d^E\beta + \mu\circ d^E(F_h\circ \beta)-(i_{\zeta_{h}^{\mu\circ \beta}}
(d^E(F_h\circ \beta)))\mu\circ \phi=0.$$ Finally, the corollary is
an immediate consequence of  Theorem \ref{main} and the relations
$$\mu\circ d^E(F_h\circ \beta)=d^V(F_h\circ \beta)\;\;\mbox{ and }\;\; \mu\circ \phi=0.$$
 \end{proof}

\begin{corollary}\label{c2}
Let $(E,\lcf\cdot,\cdot\rcf,\rho,\phi,h)$ be a {Hamiltonian system }
on the vector bundle $\tau_E:E\to Q$ on the connected manifold $Q$. Suppose that the finitely generated distribution ${\mathcal V}$ defined by ${\mathcal V}_q:=\rho_V(V_q)$ for all $q\in Q$, is a completely nonholonomic distribution. If $\beta\in \Gamma(E^*)$ is a $1$-cocycle of $E^*$, then the following statements
are equivalent:
\begin{enumerate}
\item[$(i)$] If $c:I\to Q$ is an integral curve of $R^{\mu\circ \beta}_h\in
{\mathfrak X}(Q)$ then $ \mu\circ\beta\circ c:I\to V^*$ is an integral
curve of $R_h\in {\mathfrak X}(V^*).$
\item[$(ii)$] $\beta\in \Gamma(E^*)$ satisfies the {\bf
Hamilton-Jacobi equation}
$$F_h\circ \beta=constant.$$
 \end{enumerate}
\end{corollary}


\subsection{Linear almost Poisson morphisms and Hamilton-Jacobi equation}

As we pointed out in the introduction, one important advantage of
dealing with unconstrained hamiltonian systems on Lie algebroids, or
constrained systems on skew-symmetric algebroids, is that reduction
by symmetries can be naturally handled by considering morphisms
between Lie algebroids \cite{Weinstein99} (see also
\cite{CoLeMaMaMa}), or respectively morphism between skew-symmetric
algebroids, \cite{LeMaMa}. In the following section we deal with
morphisms between skew-symmetric algebroids with a 1-cocycle, in
order to show that the Hamilton Jacobi equation is preserved by such
morphisms.

Suppose that  $\tau_E:E\to Q$ and $\tau_{\bar{E}}: \bar{E}\to
\bar{Q}$ are vector bundles over $Q$ and $\bar{Q}$, respectively.
Consider a vector bundle morphism $(\Psi,\psi)$ between $E^*$ and
$\bar{E}^*$
\[
\xymatrix{
E^*\ar[d]^{\tau_{E^*}}\ar[rr]^{\Psi}&&\bar{E}^*\ar[d]^{\tau_{\bar{E}^*}}\\
Q\ar[rr]^{\psi}&&\bar{Q}}
\]

Denote by $\wedge^k\Psi:\wedge^kE^*\to \wedge^k\bar{E}^*$ the
corresponding vector bundle morphism on $\psi:Q\to \bar{Q},$ induced
by the pair $(\Psi,\psi),$ between the vector bundles $\wedge^kE^*\to
Q$ and $\wedge^k\bar{E}^*\to \bar{Q}$. A section $\alpha\in
\Gamma(\wedge^kE^*)$ is {\it $(\Psi,\psi)$-related } with a section
$\bar{\alpha}\in \Gamma(\wedge^k\bar{E}^*)$  if
$$
\wedge^k\Psi\circ \alpha=\bar{\alpha}\circ \psi.$$

\begin{definition} Let $(E,\lcf\cdot,\cdot\rcf_E,\rho,\phi,h)$ and
$(\bar{E},\lcf\cdot,\rcf_{\bar{E}},\bar{\rho},\bar{\phi},\bar{h})$
be hamiltonian systems on $E$ and $\bar{E},$ respectively. Suppose
that  $(\Psi,\psi)$ is a vector bundle morphism
between $E^*$ and $\bar{E}^*$. Then, the pair $(\Psi,\psi)$ is said
to be a hamiltonian morphism if:
\begin{enumerate}
\item $(\Psi,\psi)$ is an almost Poisson morphism, that is,
$$\{\bar{F_1}\circ \Psi,\bar{F_2}\circ \Psi\}_{E^*}=\{\bar{F}_1,\bar{F}_2\}_{\bar{E}^*}\circ \Psi, \mbox{ for } \bar{F}_1,\bar{F}_2\in C^\infty(\bar{E}^*)$$
where $\{\cdot,\cdot\}_{E^*}$ (respectively, $\{\cdot,\cdot\}_{\bar{E}^*}$) is the linear almost Poisson bracket on $E^*$ (respectively, $\bar{E}^*$);
\item $\phi$ and $\bar{\phi}$ are $(\Psi,\psi)-$related and
\item $F_{\bar{h}}\circ \Psi =F_h.$
\end{enumerate}
\end{definition}

Now, we prove the following result
\begin{proposition}\label{morphism}
Let $(E,\lcf\cdot,\cdot\rcf_E,\rho,\phi,h)$ and $(\bar{E},\lcf\cdot,\cdot\rcf_{\bar{E}},\bar{\rho},\bar{\phi},\bar{h})$ be  hamiltonian systems on $E$ and $\bar{E},$ respectively, and $(\Psi,\psi)$ be a hamiltonian morphism between $E^*$ and $\bar{E}^*$. Then:

\begin{enumerate}
\item There exists a linear almost Poisson morphism $\widehat{\Psi}:V^*\to \bar{V}^*$ (over $\psi$) such that the following diagram is commutative

\begin{equation}\label{3.9'}
\xymatrix{
E^*\ar[d]^{\Psi}\ar[r]^{\mu}&V^*\ar[d]^{\widehat{\Psi}}\ar[r]&Q\ar[d]^\psi\\
\bar{E}^*\ar[r]^{\bar{\mu}}&\bar{V}^*\ar[r]&\bar{Q}}
\end{equation}
\item The vector fields $R_h\in {\mathfrak X}(V^*)$ and $R_{\bar{h}}\in {\mathfrak X}(\bar{V}^*)$ are $\widehat{\Psi}-$related, that is,
$$R_{\bar{h}}\circ \widehat{\Psi}=T\widehat{\Psi}\circ R_h.$$
\item

If $\alpha\in \Gamma(V^*)$ and $\bar{\alpha}\in\Gamma(V^*)$ are $(\widehat{\Psi},\psi)$-related then the vector fields $R_h^\alpha\in {\mathfrak X}(Q)$ and $R_{\bar{h}}^{\bar{\alpha}}\in {\mathfrak X}(\bar{Q})$ are $\psi$-related, that is,
$$R_{\bar{h}}^{\bar{\alpha}}\circ \psi=T\psi\circ R_h^\alpha.$$

\end{enumerate}
\end{proposition}

\begin{proof}
$(i)$ Using that $(\Psi,\psi)$ is a vector bundle morphism and the
fact that $\Psi\circ \phi=\bar{\phi}\circ \psi,$ it follows that
there exists a vector bundle morphism $\widehat{\Psi}:V^*\to
\bar{V}^*$ (over $\psi$) such that the diagram (\ref{3.9'}) is
commutative.  Moreover, since $\Psi,\mu$ and $\bar{\mu}$ are linear
almost Poisson morphisms, we deduce that $\widehat{\Psi}$ also is a
linear almost Poisson morphism.

 \medskip

 $(ii)$ The condition $F_{\bar{h}}\circ \Psi=F_h$ implies that
 $${\mathcal H}_{F_{\bar{h}}}^{\Pi_{\bar{E}^*}}\circ \Psi =T\Psi \circ {\mathcal H}_{F_h}^{\Pi_{E^*}}$$
 (note that $\Psi$ is an almost Poisson morphism). Thus, from $(i)$ and $(2.12)$, we have that
 $$R_{\bar{h}}\circ \widehat{\Psi}=T\widehat{\Psi}\circ R_h.$$

 \medskip

 $(iii)$ Using $(i)$, (\ref{Rha}) and the fact that $\widehat{\Psi}\circ \alpha=\bar{\alpha}\circ \psi,$ we conclude
 that the vector fields $R_h^\alpha$ and $R_{\bar{h}}^{\bar{\alpha}}$ are $\psi$-related.
 \end{proof}

 From Proposition \ref{morphism}, we deduce that following result

\begin{theorem} \label{theo:morphism}
Let $(E,\lcf\cdot,\cdot\rcf_E,\rho,\phi,h)$ and
$(\bar{E},\lcf\cdot,\cdot\rcf_{\bar{E}},\bar{\rho},\bar{\phi},\bar{h})$
be hamiltonian systems on $E$ and $\bar{E}$, respectively, and
$(\Psi,\psi)$ be a hamiltonian morphism between $E^*$ and
$\bar{E}^*.$ If  $\alpha\in \Gamma(V^*)$ satisfies the Hamilton-Jacobi equation for $h$,  $\bar{\alpha}\in
\Gamma(\bar{V}^*)$ is $(\Psi,\psi)$-related with $\alpha$ and $\psi$ is a surjective map then
$\bar{\alpha}$ satisfies the
Hamilton-Jacobi equation for $\bar{h}$.

\end{theorem}

\begin{remark}
{\rm Let $(E,\lcf\cdot,\cdot\rcf_E,\rho,\phi,h)$ and
$(\bar{E},\lcf\cdot,\cdot\rcf_{\bar{E}},\bar{\rho},\bar{\phi},\bar{h})$
be hamiltonian systems on $E$ and $\bar{E}$, respectively, and
$(\Psi,\psi)$ be a hamiltonian morphism between $E^*$ and
$\bar{E}^*.$ Suppose that $\bar\alpha\in \Gamma(E^*)$ is a $1$-cocycle of $\bar{E}^*$ such that
$F_{\bar{h}}\circ \bar{\alpha}= constant.$ Then, if $\alpha\in\Gamma(E^*)$ is a $1$-cocycle of $E^*$ which is
$(\Psi,\psi)$-related with $\bar{\alpha}$, it is clear that $F_h\circ \alpha=constant$  and, therefore, $\alpha$
is a solution of the Hamilton-Jacobi equation for $h.$}
\end{remark}


\section{Examples}
In this section we will apply our theory to two type of mechanical
systems: uncontrained mechanical systems with  a dissipative
character (with linear external forces or time-dependent systems)
and nonholonomic mechanical systems subjected to  affine
constraints. In the last part of the section we will discuss the
case of a nonholonomic mechanical system with external linear
forces.

\subsection{Uncontrained mechanical systems with a dissipative character}

\subsubsection{Mechanical systems on Lie
algebroids with  linear  external forces}\label{Ex1:freeSkew}

(See Example  \ref{ejemplo}). Let us consider a  Lie  algebroid
structure (or more generally a skew-symmetric algebroid)
$(\lcf\cdot,\cdot\rcf,\rho)$ on a vector bundle $\tau_{\bar{E}}:
{\bar{E}}\to Q$ and a homomorphism of vector bundles $F:{\bar{E}}\to
{\bar{E}}$ (over the identity of $Q$). With these ingredients, it is
induced on the vector bundle $\tau_{\mathbb{R}\times
{\bar{E}}}:\mathbb{R}\times \bar{E}\to Q$, a skew-symmetric
algebroid structure $(E:= \mathbb{R}\times \bar{E},
\lcf\cdot,\cdot\rcf_{\mathbb{R}\times
\bar{E}},\rho_{\mathbb{R}\times \bar{E}})$ such that $(1,0)\in
\Gamma(\mathbb{R}\times \bar{E}^*)\cong C^\infty(M)\times
\Gamma(\bar{E}^*)$ is a $1$-cocycle.

Let $H:\bar{E}^*\to {\mathbb R}$ be a differentiable function ({\it Hamiltonian function}) on $\bar{E}^*.$ Denote by  $h: \bar{E}^*\to {\mathbb R}\times \bar{E}^*$ the induced section of $\mu=pr_2:{\mathbb R}\times \bar{E}^*\to \bar{E}^*$ by $H$, i.e.,
$$ h(\beta_q)=(-H(\beta_q),\beta_q),\;\;\;\;\mbox{ for all $q\in Q$ and $\beta_q\in \bar{E}_q^*.$}$$

The vector field $R_h$ on $\bar{E}^*$ is just ${\mathcal H}_H^{\Pi_{\bar{E}^*}}-(F^*)^\vee$ (see (\ref{1.14})).  Moreover,
$$R_{h}^\alpha=T\tau_{\bar{E}^*}\circ {\mathcal H}_H^{\Pi_{\bar{E}^*}} \circ \alpha.$$

On the other hand, for each $\alpha\in \Gamma(\bar{E}^*)$ one may define a section, $\zeta_{H}^{\alpha},$ of $ \bar{E}$ as follows
$$\beta(\zeta_{H}^{\alpha})=\beta^\vee(H)\circ \alpha,$$
for  $\beta\in \Gamma(\bar{E}^*)$. Then, under the identification $\Gamma({\mathbb R}\times \bar{E})\cong C^\infty(Q)\times \Gamma(\bar{E})$, $\zeta_{h}^{\alpha}$ is just $(1,\zeta_{H}^{\alpha}).$
Thus, using Corollary \ref{c1} we conclude the following result

\begin{corollary}\label{c3}
Let $(\bar{E},\lcf\cdot,\cdot\rcf,\rho)$ be a Lie algebroid (or more
generally a skew-symmetric algebroid) with Hamiltonian function
$H:\bar{E}^*\to {\mathbb R}$, and  $F:{\bar{E}}\to {\bar{E}}$ be a
vector bundle homomorphism. If $\alpha\in \Gamma(\bar{E}^*)$, the
following statements are equivalent:
\begin{enumerate}
\item[$(i)$] If $c:I\to Q$ is an integral curve of $T\tau_{\bar{E}^*}\circ {\mathcal H}_H^{\Pi_{\bar{E}^*}}\circ \alpha\in
{\mathfrak X}(Q)$ then $ \alpha\circ c:I\to E^*$ is an integral
curve of ${\mathcal H}_H^{\Pi_{\bar{E}^*}}-(F^*)^\vee.$
\item[$(ii)$] $\alpha\in \Gamma(\bar{E}^*)$ satisfies the {\bf
Hamilton-Jacobi equation}, i.e.,
$$ i_{\zeta_{H}^{\alpha}}d^{\bar{E}}\alpha + d^{\bar{E}}(H\circ \alpha) + F^*(\alpha)=0.$$
 \end{enumerate}
\end{corollary}

\begin{remark}{\rm
\begin{enumerate}
\item
When $\alpha$ is a $1$-cocycle  and $F\equiv 0$, we recover the result of \cite{LeMaMa}. Applications of this result to
nonholonomic mechanical systems subjected to linear constraints were
discussed there. Note that in this case the dissipative term is zero.
\item In the particular case when $\bar{E}$ is the standard Lie algebroid
$\tau_{TQ}:TQ\to Q$, then, using well-known results (see, for
instance, \cite{LR}), we deduce that there exists a one-to-one
correspondence between the vector bundle morphisms $F:TQ\to T^*Q$
(over the identity of $Q$) and the semi-basic $1$-forms on $TQ$
which are homogeneous of degree $1$. A semi-basic $1$-form
$\beta:TQ\to T^*(TQ)$ on $TQ$ is said to be homogeneous of degree
$1$ if ${\mathcal L}_\Delta\beta=\beta$, where $\Delta$ is the
Liouville vector field on $TQ$. Indeed, if
\[
F(q^i,\dot{q}^i)=(q^i,F_j^i(q)\dot{q}^j)
\]
then the corresponding $1$-form $\beta$ on $TQ$ is given by
\[
\beta=(F_j^i(q)\dot{q}^j)dq^i.
\]
Using this result and Corollary \ref{c3}, we deduce Theorem 3.4
in \cite{ILM} for the particular case when the semi-basic $1$-form
$\beta$ on $TQ$ is homogeneous of degree $1$.
\end{enumerate}}
\end{remark}

\begin{example}\label{Classical}{\emph{Standard mechanical systems. }}
{\rm
 Let $\bar{E}$ be  the standard Lie algebroid $\tau_{TQ}:TQ\to Q$.  In this case the differential $d^{\bar{E}}$ is the standard differential, $d$,  on $Q.$
 Suppose that $F\equiv 0$ and that $H:T^*Q\to \R$ is a hamiltonian function.
 If $\alpha$ is a $1$-form on $Q$ then the Hamilton-Jacobi equation is
 \begin{equation}\label{HJ-est}
 i_{X_{H}^\alpha}d\alpha + d (H\circ \alpha)=0,
 \end{equation}
 where $X_H^\alpha$ is the vector field on $Q$ defined by $X_H^\alpha(\beta)=\beta^\vee(H)\circ \alpha$, for all $\beta\in\Omega^1(Q).$

 If $Q$ is connected and $S:Q\to \R$ is a function on $Q$, using the $1$-form $\alpha=dS$, we obtain the standard Hamilton-Jacobi equation on $Q$, i.e.,
 $$H\circ dS=constant.$$

  On the other hand,  let ${\mathcal G}$ be a
riemannian metric  on a n-dimensional manifold $Q$, i.e, a
positive-definite  symmetric $(0,2)$-tensor on $Q$. The metric
${\mathcal G}$ induces the musical isomorphisms
\begin{eqnarray*}
&&\flat_{\mathcal G}: {\mathfrak X}(Q)\longrightarrow \Omega^1(Q), \qquad \flat_{\mathcal G}(X)(Y)={\mathcal G}(X,Y),\\
&&\sharp_{\mathcal G}:\Omega^1(Q)\longrightarrow{\mathfrak X}(Q),\qquad\sharp_{\mathcal G}(\alpha)=\flat_{\mathcal G}^{-1}(\alpha)
\end{eqnarray*}
where $X, Y\in {\mathfrak X}(Q)$ and $\alpha\in \Omega^1(Q)$.

Associated with the metric ${\mathcal G}$ there is an affine
connection $\nabla^{\mathcal G}$, the \emph{Levi-Civita
connection}, determined by:
$$
\begin{array}{l}
 [X,Y]= \nabla^{\mathcal G}_X Y-\nabla^{\mathcal G}_Y X \quad \hbox{   (symmetry)}\\
 X( {\mathcal G}(Y,Z))={\mathcal G}(\nabla^{\mathcal G}_X Y, Z)+{\mathcal G}(Y, \nabla^{\mathcal G}_X Z)\quad \hbox{   (metricity)}\; ,
 \end{array}
$$
for every $X, Y, Z\in {\mathfrak X}(Q)$.

Considering a vector field $X\in {\mathfrak X} (Q)$ and the associated 1-form $\alpha=\flat_{\mathcal G}(X),$
we will analyze the meaning of the Hamilton-Jacobi Equation (\ref{HJ-est}) for the Hamiltonian $H: T^*Q\to \R$ defined by $H(\eta_q)=\frac{1}{2}{\mathcal G}^* (\eta_q, \eta_q),$ where ${\mathcal G}^*$ is the induced metric on $T^*Q$ and $\eta_q\in T^*_q Q$.  First, we observe that the section $\zeta_H^\alpha$ of $TQ$ is just the vector field $X$. Then,  for $Y\in {\mathfrak X}(Q)$
\begin{eqnarray*}
d \alpha( X, Y)+Y(H\circ \alpha)
                    &=&d (\flat_{\mathcal G}(X))(X,Y)+\frac{1}{2}Y({\mathcal G}(X,X))\\
                    &=&  X ({\mathcal G}(X,Y))-\frac{1}{2}Y({\mathcal G}(X,X))-{\mathcal G}(X,[X, Y])\\
                    &=& {\mathcal G}(\nabla^{\mathcal G}_X X,Y).
\end{eqnarray*}

Therefore, the Hamilton-Jacobi equation (\ref{HJ-est}) for the case
of a Hamiltonian defined by a riemmanian metric is equivalent to the
condition for auto-parallelism of vector fields, that is, vector
fields $X\in {\mathfrak X}(Q)$ such that $\nabla^{\mathcal G}_X
X=0$.

Thus, if we have a vector field $X$ which satisfies the auto-parallelism condition, each integral curve $c:I\to Q$ (which is a geodesic) induces a solution of the Hamilton equations of the mechanical system, which is just
$$\flat_{\mathcal G}(X)\circ c:I\to T^*Q.$$
}
\end{example}


\begin{example}{\emph{The test particle under the gravitational interaction of two masses. }}
{\rm Consider the problem of the motion of a  particle moving under the
gravitational effect of two masses $m_1$ and $m_2$, which in turn
move in circular orbits about their common center of mass and are
not influenced by the motion of the particle  (classical planar
circular restricted three-body problem). Take a coordinate system
rotating about the common center of mass with the same frequency as
the two masses so that both of them lie on the $x$-axis with
coordinates $(-\mu_2, 0)$ and $(\mu_1, 0)$, where $\mu_i=
\frac{m_i}{m_1+m_2}$ (see \cite{KrMa,Murray}).

The system is described by the Lagrangian function:
\[
L(x,y,\dot{x}, \dot{y})=\frac{1}{2}(\dot{x}-y)^2+\frac{1}{2} (\dot{y}+x)^2-\frac{\mu_1}{r_1}-\frac{\mu_2}{r_2},
\]
where
$r^2_1=(x+\mu_2)^2+y^2$ and $r^2_2=(x-\mu_1)^2+y^2$.

The equations of motion adding a drag force $\tilde{F}=(\tilde{F}_1(x,y,\dot{x}, \dot{y}), \tilde{F}_2(x,y,\dot{x}, \dot{y}))$ are (see \cite{Murray}):
\begin{eqnarray*}
\ddot{x}-2\dot{y}-x&=&- \frac{\partial U}{\partial x}-\tilde{F}_1,\\
\ddot{y}+2\dot{x}-y&=&-\frac{\partial U}{\partial y}-\tilde{F}_2,
\end{eqnarray*}
where $U(x,y)=\frac{\mu_1}{r_1}+\frac{\mu_2}{r_2}$ and $\tilde{F}: T\R^2\to T^*\R^2$.

Now, we will describe this system in our geometric
framework. Consider the homomorphism $F: T\R^2\to T\R^2$ given by
\begin{eqnarray*}
F(\frac{\partial}{\partial x})&=&F^1_1(x,y)\frac{\partial}{\partial x}+F^2_1(x,y)\frac{\partial}{\partial y},\\
F(\frac{\partial}{\partial y})&=&F^1_2(x,y)\frac{\partial}{\partial x}+F^2_2(x,y)\frac{\partial}{\partial y},\
\end{eqnarray*}
where $F^a_b \in C^{\infty}(\R^2)$. Then, on the vector bundle
$\tau: \R \times T\R^2 \rightarrow \R^2$ it is induced a
(transitive) skew-symmetric algebroid structure described, in the local basis
$\{e_0=(1,0), e_1=(0, \frac{\partial}{\partial x}), e_2=(0,
\frac{\partial}{\partial y}) \} $, as follows
$$\begin{array}{c}
\lcf (1,0), (0,\displaystyle\frac{\partial}{\partial x})\rcf_{\R\times T\R^2}=(0, -F(\displaystyle\frac{\partial}{\partial x})),\;\;\;
\lcf (1,0), (0,\displaystyle\frac{\partial}{\partial y})\rcf_{\R\times T\R^2}=(0, -F\displaystyle(\frac{\partial}{\partial y})), \\[8pt]
\rho_{\R\times T\R^2}(1,0)=0,\;\;\; \rho_{\R\times T\R^2}(0,\displaystyle\frac{\partial}{\partial x})=\displaystyle\frac{\partial}{\partial x},\;\;\; \rho_{\R\times T\R^2}(0,\displaystyle\frac{\partial}{\partial
y})=\displaystyle\frac{\partial}{\partial y}.
\end{array}$$

Therefore, $\C_{01}^1=-F^1_1$, $\C_{01}^2=-F^2_1$,
$\C_{02}^1=-F^1_2$, $\C_{02}^2=-F^2_2$, $\rho^1_1=1$ and
$\rho^2_2=1$.

Note that the homomorphism $F$ generates a drag force $\tilde{F}$ of
the type \begin{eqnarray*}
\tilde{F}(x, y, \dot{x}, \dot{y})&=& \left(F^1_1(x,y)\frac{\partial L}{\partial
\dot{x}}+F^2_1(x,y)\frac{\partial L}{\partial \dot{y}},
F^1_2(x,y)\frac{\partial L}{\partial \dot{x}}+F^2_2(x,y)\frac{\partial L}{\partial \dot{y}} \right)\\
&=& \left(F^1_1(x,y)(\dot{x}-y)+F^2_1(x,y)(\dot{y}+x),
F^1_2(x,y)(\dot{x}-y)+F^2_2(x,y)(\dot{y}+x)\right).
\end{eqnarray*}

On the dual bundle $\R\times T^*\R^2$ we have a hamiltonian
function:
\begin{equation}
H(x, y, p_x, p_y)=\frac{1}{2}p_x^2+\frac{1}{2}p_y^2+yp_x-xp_y+U(x,y)
\label{eq:HamTParticle} \end{equation} and the corresponding
Hamilton's equations are now:
\begin{eqnarray*}
\dot{x}&=& p_x+y,\\
\dot{y}&=&p_y-x,\\
\dot{p}_x&=&p_y-\frac{\partial U}{\partial x}-F_1^1(x,y)p_x-F_1^2(x,y)p_y=-\frac{\partial H}{\partial x}-F_1^1(x,y)p_x-F_1^2(x,y)p_y,\\
\dot{p}_y&=&-p_x-\frac{\partial U}{\partial y}-F_2^1(x,y)p_x-F_2^2(x,y)p_y=-\frac{\partial H}{\partial y}-F_2^1(x,y)p_x-F_2^2(x,y)p_y.
\end{eqnarray*}

Consider a section $\alpha\in \Gamma(T^*\R^2)$ where $\alpha=\alpha_1\, dx+\alpha_2\, dy$. Then,
$$\zeta_{H}^\alpha=(\frac{\partial H}{\partial p_x} \circ \alpha)\frac{\partial }{\partial x}+
(\frac{\partial H}{\partial p_y} \circ \alpha)\frac{\partial }{\partial y}$$

Thus, Hamilton-Jacobi equation is equivalent to
\begin{eqnarray*}
\frac{\partial H}{\partial x}\circ \alpha+(\frac{\partial H}{\partial p_x}\circ \alpha)
\frac{\partial \alpha_1}{\partial x}+(\frac{\partial H}{\partial p_y}\circ \alpha)\frac{\partial \alpha_1}{\partial y}+\alpha_i F^i_1 &=&0\\
\frac{\partial H}{\partial y}\circ \alpha+(\frac{\partial H}{\partial
p_x}\circ \alpha)\frac{\partial \alpha_2}{\partial x}+(\frac{\partial H}{\partial
p_y}\circ \alpha)\frac{\partial \alpha_2}{\partial y} + \alpha_i F^i_2&=&0\, .
\end{eqnarray*}

For a hamiltonian function given by (\ref{eq:HamTParticle}), the
last two equations can be written as:
\begin{eqnarray*}
\frac{\partial U}{\partial x}-\alpha_2+(\alpha_1+y)\frac{\partial
\alpha_1}{\partial x}+
(\alpha_2-x)\frac{\partial \alpha_1}{\partial y} + \alpha_i F^i_1&=&0\\
\frac{\partial U}{\partial y}+\alpha_1+(\alpha_1+y)\frac{\partial
\alpha_2}{\partial x}+ (\alpha_2-x)\frac{\partial \alpha_2}{\partial
y} + \alpha_i F^i_2&=&0\, .
\end{eqnarray*}

An interesting case \cite{Murray} is when the drag force is
$$\tilde F(x,y,\dot x,\dot y) = \left(k(x,y)(\dot{x}-y),
k(x,y)(\dot{y}+x)\right),$$ with $k\in C^{\infty}(\R^2)$. In this
case, the homomorphism is $F(X)=k(x,y)X$ with $X\in T_{(x,y)}\R^2$.
Thus, the equations of motion are
\begin{eqnarray*}
\dot{x}&=& p_x+y,\\
\dot{y}&=&p_y-x,\\
\dot{p}_x&=&p_y-\frac{\partial U}{\partial x}-k(x,y)p_x,\\
\dot{p}_y&=&-p_x-\frac{\partial U}{\partial y}-k(x,y)p_y.
\end{eqnarray*}

In this particular case the linear almost Poisson tensor on $\R\times T^*\R^2$ is given by
$$
\Pi_{\R\times T^*\R^2}=\frac{\partial }{\partial x}\wedge
\frac{\partial}{\partial p_x}+ \frac{\partial }{\partial y}\wedge
\frac{\partial}{\partial p_y}
+k(x,y)p_x\frac{\partial }{\partial p_0}\wedge \frac{\partial }{\partial p_x}\\
+ k(x,y)p_y\frac{\partial }{\partial p_0}\wedge \frac{\partial
}{\partial p_y}
$$
where $(p_0, x, y, p_x, p_y)$ are standard coordinates on $\R\times
T^*\R^2$.

Now, if  the function $k$ is constant and we choose a section
$\alpha = dS$ where $S: \R^2\to \R$ is an
arbitrary function,
 the Hamilton-Jacobi equation is $$ d(H\circ \alpha) + k dS =0,
$$ which is equivalent to the suggestive equation:
 \[
 kS+H\circ dS=\hbox{constant}
 \]
 or, in other words,
 \[
 kS(x,y)+H(x,y, \frac{\partial S}{\partial x}, \frac{\partial S}{\partial
y})=\hbox{constant}.
 \]
In particular, for the hamiltonian function given by
(\ref{eq:HamTParticle}), we obtain the following partial
differential equation:
 \[
 kS(x,y)+ \frac{1}{2}\left(\frac{\partial S}{\partial x}\right)^2 +
 \frac{1}{2}\left(\frac{\partial S}{\partial y}\right)^2+y\frac{\partial S}{\partial x}-x\frac{\partial S}{\partial y}+U(x,y)=\hbox{constant}.
\]

Note that, this equation is the Hamilton-Jacobi equation as stated
in Corollary \ref{c2} for the cocycle $(kS,dS) \in C^\infty(\R^2)
\times \Omega^1(\R^2) \simeq \Gamma((\R \times T\R^2)^*)$ when we
consider the skew-symmetric algebroid $\tau:\R \times T\R^2
\rightarrow \R^2$

}

\end{example}

\begin{example}{\emph{Hamilton-Jacobi equation for a particle on a vertical cylinder in a uniform gravitational field with friction. }}
{\rm As another example we consider a particle of mass $m$
constrained to move on  a cylinder of radius $r$ in a uniform
gravitational field of strength $g$ and assume also the existence of
a frictional force acting on the system.

The Hamiltonian $H: T^*(\R\times S^1)\to \R$ is:
\[
H(x,\theta,p_x, p_\theta)=
\frac{p_{x}^2}{2m}+\frac{p_{\theta}^2}{2mr^2}+mgx\; .
\]
The frictional force is modeled in our setting by   the homomorphism
$F: T(\R\times S^1)\to T(\R\times S^1)$
given by
\begin{eqnarray*}
F(\frac{\partial}{\partial x})&=&K_1\frac{\partial}{\partial x},\\
F(\frac{\partial}{\partial \theta})&=&K_2\frac{\partial}{\partial
\theta}\; ,
\end{eqnarray*}
with $(K_1, K_2)\in\R^2$.

The corresponding equations of motion are
\begin{eqnarray*}
m\dot{x}&=&p_x\\
mr^2\dot{\theta}&=& p_\theta\\
\dot{p}_x&=&-mg-K_1 p_x\\
\dot{p}_\theta&=&-K_2 p_\theta\; .
\end{eqnarray*}

In this case, we may consider the skew-symmetric algebroid
$\tau:\R\times T(\R\times S^1) \to \R \times S^1$ associated with
the above homomorphism $F: T(\R\times S^1)\to T(\R\times S^1)$
defined as in (\ref{1Ex}). For this skew-symmetric algebroid we have
that $\phi=(1,0) \in \Gamma(\R\times T^* (\R\times
S^1))=C^\infty(\R\times S^1)\times \Omega^1(\R\times S^1)$ is a
1-cocycle and $V=\widehat{\phi}^{-1}(0)=T(\R\times S^1)$.

Let us consider a $1$-form $\alpha \in \Gamma(T^*(\R\times S^1))$.
If locally  $\alpha$ is given by $$\alpha = \alpha_x dx +
\alpha_\theta d\theta$$ with $\alpha_x, \alpha_\theta \in
C^\infty(\R \times S^1)$, then the local expression of the vector
field $\zeta_H^\alpha$ on $\R\times S^1$ is
\[
\zeta_H^\alpha=\frac{\alpha_x}{m}\frac{\partial }{\partial x} + \frac{\alpha_\theta}{m}\frac{\partial }{\partial \theta}.
\]
Moreover, the Hamilton-Jacobi equation for $\alpha\in \Omega^1(\R\times S^1)$ is
\[
i_{\zeta_H^\alpha} d\alpha + d(H\circ \alpha) + F^*\alpha=0,\]
where $d$ is the standard differential and $F^*\alpha$ is the pullback of $\alpha$ by $F$ (see Example \ref{ejemplo}).
In local coordinates this equation becomes
\begin{eqnarray*}
\frac{\alpha_x}{m} \frac{\partial \alpha_x}{\partial x} + \frac{\alpha_\theta}{m r^2}\frac{\partial \alpha_x}{\partial \theta} + m g + K_1 \alpha_x &=&0,\\
\frac{\alpha_x}{m}\frac{\partial \alpha_\theta}{\partial x} + \frac{\alpha_\theta}{m r^2}\frac{\partial \alpha_\theta}{\partial \theta}  + K_2 \alpha_\theta &=&0.
\end{eqnarray*}

In the particular case when $\alpha=dS$ with $S$  a function given by $S(x, \theta)=S^{(1)}(x)+S^{(2)}(\theta)$, we have that the corresponding Hamilton-Jacobi equation is
\begin{eqnarray*}
K_1\frac{dS^{(1)}}{d x}+mg+\frac{1}{m}\frac{d S^{(1)}}{d x}\frac{d^2 S^{(1)}}{d x^2}&=&0,\\
K_2\frac{dS^{(2)}}{d \theta}+\frac{1}{mr^2}\frac{d
S^{(2)}}{d \theta}\frac{d^2 S^{(2)}}{d \theta^2}&=&0\; .
\end{eqnarray*}

Solving the equation we obtain that
\begin{eqnarray*}
S^{(2)}(\theta) &=&0 \hbox{     or     }  S^{(2)}(\theta)=-\frac{K_2mr^2}{2} \theta^2+C_1\\
S^{(1)}(x)&=& \frac{-gm-gm \mathbf{W}\left[\displaystyle{
-\frac{e^{-1+\frac{K_1^2x}{g}-\frac{K_1^2
C_2}{gm}}}{gm}}\right]}{K_1} \hbox{   if $K_1\not=0$ or    }
\\
S^{(1)}(x)&=&\pm\sqrt{2}\sqrt{-gm^2 x+C_2} \hbox{   if $K_1=0$}\; ,
\end{eqnarray*}
where $\mathbf{W}$ is the Lambert W-function (the inverse function
of $f(v)=ve^v$) and $C_1$ and $C_2$ are constants.

\begin{figure}\label{Fig1}

 \includegraphics[height=5.5cm]{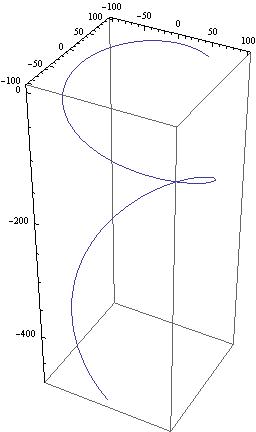}\qquad \qquad \includegraphics[height=5.5cm]{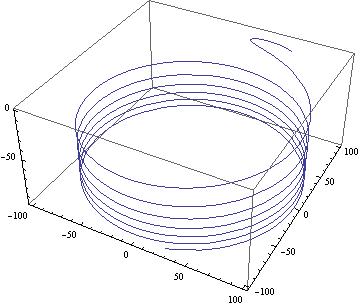}\\
 \caption{Comparison of a free trajectory (without friction), on the left,  and a trajectory with friction, on the right}
\end{figure}

In Figure $1$, we compare the trajectory of the particle for the
free problem and the trajectory for the problem with friction.}

\end{example}


\subsubsection{Unconstrained
mechanical systems on  Lie algebroids with a
1-cocycle}\label{section3.2}

Let $(\lcf\cdot,\cdot\rcf,\rho)$ be a Lie algebroid structure (or
more generally a skew-symmetric algebroid) on a vector bundle
$\tau_E:E\to Q$ and $\phi\in \Gamma(E^*)$ be a $1$-cocycle such that
$\phi(q)\not=0,$ for all $q\in Q.$ Denote by ${\mathcal A}$
(respectively, $V$) the affine (respectively, vector)  subbundle of
$E$  given by ${\mathcal A}=\widehat{\phi}^{-1}(1)$ (respectively,
$V=\widehat{\phi}^{-1}(0)$).

In addition, we endow the vector bundle with a bundle metric
${\mathcal G}:E\times_Q E\to \mathbb{R}$ on $E$.  Denote by
$\flat_{\mathcal G}:E\to E^*$ the isomorphism of vector bundles
induced by ${\mathcal G}.$  Consider the section ${\mathcal X}$ of
$E$ defined as follows
$${\mathcal X}=\flat_{\mathcal G}^{-1}\circ \phi. $$

We will suppose, without loss of generality, that ${\mathcal
G}({\mathcal X},{\mathcal X})=1.$ Thus, $\phi({\mathcal X})=1$ and ${\mathcal X}$ is a section of the affine bundle $\tau_{\mathcal A }:{\mathcal A}\to Q.$  On the other hand, ${\mathcal G}({\mathcal X}(q),v)=0,$ for all $v\in
  V_q$, therefore $E_q=<{\mathcal X}(q)> \oplus V_q,$ for all $q\in Q.$

Now, let us consider the hamiltonian section $h:V^*\to E^*$ of the
AV-bundle $\mu:E^*\to V^*$ characterized by
\[
h(\eta_q)(v_q)=\eta_q(v_q),\;\;\; h(\eta_q)({\mathcal
X}_q)=-H(\eta_q),\;\;\mbox{ for all } q\in Q,\;\; \eta_q\in V_q^*
\mbox{ and } v_q\in V_q,
\]
where $H:V^*\to \mathbb{R}$  is the function
$$
H(\eta_q)=\frac{1}{2}{\mathcal G}^*_V(\eta_q,\eta_q) + {\mathbb
V}(q),$$ with ${\mathcal G}^*_V:V^*\times V^*\to \mathbb{R}$ the
bundle metric induced by ${\mathcal G}$ on $V^*$ and ${\mathbb
V}:Q\to \mathbb{R}$ a real $C^\infty$-function on $Q$. Then, the
function $F_h:E^*\to \mathbb{R}$ associated with the section $h$ is
just $F_h=H\circ \mu +\widehat{\mathcal X}.$

Let $(q^i)$ be a system of local coordinates for $Q$ and consider an
orthonormal local basis $\{e_0,e_a\}$ of $\Gamma(E)$ with
$e_0={\mathcal X}$. Denote by $(q^i,p_0,p_a)$ the local coordinates
on $E^*$ with respect to the dual basis of $\{e_0,e_a\}$.

The local expression of the hamiltonian section  $h\in
\Gamma(\mu)$ is the following
\[
\begin{array}{rcl}
h(q^i,p_a)&=&(q^i,-H=\displaystyle\displaystyle-\frac{1}{2}(p_a)^2-
{\mathbb V}(q), p_a).
\end{array}
\]
The integral curves of the hamiltonian vector field $R_h\in
{\mathfrak X}(V^*)$ are the solutions of the Hamilton equations
\[
\begin{array}{rcl}
\displaystyle\frac{dq^i}{dt}&=& \rho_0^i + \rho_a^ip_a\\[8pt]
\displaystyle\frac{dp_b}{dt}&=&-\rho_b^i\displaystyle\frac{\partial
{\mathbb V} }{\partial q^i} + (\C_{0b}^c + \C_{ab}^cp_a)p_c.
\end{array}
\]

For this mechanical system the dissipative term has the local
expression
\[
\{H\circ \mu,F_h\}=\rho_0^i\frac{\partial {\mathbb V}}{\partial
q^i}+\C^c_{0b}p_cp_b.\]

Let $\alpha$ be a section of $V^*$. Then, the section $\zeta_h^\alpha$ of $E$ is given by
$$
\zeta_{h}^\alpha=i_V\circ \zeta_H^\alpha + {\mathcal X},$$
where $\zeta_H^\alpha$ es the section of $V$ defined by $\beta(\zeta_H^\alpha)=\beta^\vee(H)\circ \alpha,$ for all $\beta\in \Gamma(V^*)$.

Thus, using Theorem \ref{main} we deduce the following corollary

\begin{corollary}\label{c4.2}
Let $\alpha$ be a  section of $V^*$. Then, the following statements
are equivalent:
\begin{enumerate}
\item[$(i)$] If $c:I\to Q$ is an integral curve of $R_h^\alpha=
T\tau_{V^*}\circ R_h\circ \alpha \in {\mathfrak X}(Q)$, then
$\alpha \circ c:I\to V^*$ is an integral curve of $R_h\in
{\mathfrak X}(V^*).$
\item[$(ii)$] $\alpha\in \Gamma(V^*)$ satisfies the {
Hamilton-Jacobi equation}:
$$i_{\zeta_H^{\alpha}}d^V\alpha +\mu\circ i_{\mathcal X}d^E (h\circ \alpha)=0.$$
\end{enumerate}
\end{corollary}

\begin{remark}{\rm If $\beta$ is a  section of $E^*$, the Hamilton-Jacobi equation for $\mu\circ \beta$ is (see Corollary \ref{c1})
$$i_{\zeta_H^{\mu\circ \beta}}d^V(\mu\circ \beta) +\mu\circ i_{\mathcal X}d^E (\beta) + d^V(H\circ \mu\circ\beta + \widehat{\mathcal X}\circ \beta)=0.$$

If $\beta$ is $1$-cocycle on $E$, from Corollary \ref{c2}, then the Hamilton-Jacobi equation is
\begin{equation}\label{Remark48}
d^V(H\circ \mu\circ \beta + \widehat{\mathcal X}\circ \beta)=0.
\end{equation}
Therefore, Corollary \ref{c4.2} is a generalization of the main result of \cite{MaSo}
(see Theorem 3 in \cite{MaSo}).  In such a paper the
authors obtain a Hamilton-Jacobi equation for mechanical systems on
Lie affgebroids with this extra hypothesis on $\beta$.

If, additionally, $V$ is a transitive Lie algebroid (that is, $\rho_V({V}_q)=T_qQ$, for all $q\in Q$) and $Q$ is connected,
we have that the Eq. (\ref{Remark48}) may be rewritten as follows
\[
H\circ\mu\circ \beta + \widehat{\mathcal X}\circ \beta=constant.
\]
}
\end{remark}
\begin{example}{\emph{Time dependent classical systems. }}
{\rm Let $\pi:Q\to \R$ be a fibration on a manifold $Q$ and $\eta$
the $1$-form on $Q$ given by $\eta=\pi^*(dt)$, where $t$ is the
standard coordinate on $\R$. Consider the standard Lie algebroid on
$TQ$. Then $\eta$ is a $1$-cocycle for it and the affine bundle
${\mathcal A}=\widehat{\eta}^{-1}(1)=\{v\in TQ/\eta(v)=1\}\to Q$ may
be identified with  the $1$-jet bundle $J^1\pi$ of local sections of
$\pi.$ Note that the associated vector bundle
$V=\widehat{\eta}^{-1}(0)$ is just the vertical bundle of $\pi$
$$V\pi=\{ v\in TQ/\eta(v)=0\}.$$

Now, we take $h:V^*\pi\to T^*Q$ a hamiltonian section of $\mu:T^*Q\to V^*\pi.$ If the local expression of $h$ is
\[
h(t,q^i,p_i)=(t,q^i,-H(t,q^i,p_i),p_i)
\]
then the associated hamiltonian vector field $R_h$ on $V^*\pi$ is given by
\[
R_h=\frac{\partial}{\partial  t} + \frac{\partial H}{\partial
p_i}\frac{\partial}{\partial  q^i}-\frac{\partial H}{\partial
q^i}\frac{\partial }{\partial  p_i}.\]

Thus, the Hamilton equations are just the time dependent classical
Hamilton equation for $h$
\[
\frac{dq^i}{dt}=\frac{\partial H}{\partial  p_i},\;\;\; \frac{dp_i}{dt}=-\frac{\partial H}{\partial  q^i}.
\]
Now, consider a section $\alpha$  of the vector bundle  $V^*\pi\to
Q$. Then, $\zeta_h^\alpha = \widetilde{dF_h}$ is a vector field on
$Q$ defined by
\[
\beta(\zeta_h^\alpha)=\beta^\vee(F_h)\circ h\circ \alpha,\;\;\;\mbox{ for } \beta\in \Omega^1(E).\]
The Hamilton-Jacobi equation is
\[
(i_{\zeta_h^\alpha} d(h\circ \alpha))_{|V\pi}=0.
\]
In the case when $\alpha$ is a closed $1$-form on $Q$ the Hamilton-Jacobi equation
may be rewritten as
\[
(d(F_h\circ \alpha))_{|V\pi}=d^{V\pi}(F_h\circ \alpha)=0,
\]
i.e., $F_h^\alpha=F_h\circ \alpha$ is constant on the fibers of $\pi$.

Finally, we analyze the case when $\pi$ is trivial, that is, $Q=\R\times P$ with $P$ a connected manifold and $\pi$ is the projection on the first factor. Then, $V\pi=\R\times TP$ and the section $h$ may be identified with a time dependent Hamiltonian function $H:\R\times T^*P\to \R.$ If $\alpha=dW$, with $W:Q\to \R$ a real function on $Q$, then
\[(F_h\circ \alpha)(t,q)=\frac{\partial W }{\partial t}_{|t} + H(t,dW_t(q))\]
with $(t,q)\in \R\times P.$ Here $W_t:P\to \R$ is the real function defined by $W_t(p)=W(t,p).$
In this case the local expression of the Hamilton-Jacobi equation is
\[
\frac{\partial W}{\partial t} + H(t,q^i,\frac{\partial W}{\partial q^i})=\mbox{ constant on $P$,}
\]
i.e., the time dependent classical Hamilton-Jacobi equation \cite{AbMa}.}
\end{example}

\subsection{Nonholonomic mechanical systems with   affine constraints}

Let  $(\lcf\cdot,\cdot\rcf,\rho)$ be a Lie algebroid structure on a
vector bundle $\tau_E:E\to Q$.

A {\it mechanical system subjected to affine nonholonomic constraints on $E$ } consists of
\begin{enumerate}
\item a vector subbundle $\tau_U:U\to Q$ of $E,$
\item a bundle metric ${\mathcal G}:E\times_Q E\to \R$ on $E,$
\item a function ${\mathbb V}:Q\to \R$ on $Q$
\item and a section $X_0\in \Gamma(E)$ such that ${\mathcal P}(X_0)=0,$ where ${\mathcal P}:E=U\oplus U^\perp \to U$ is the orthogonal projector defined by the metric ${\mathcal G}.$
\end{enumerate}

Then, one  may consider an affine bundle $\tau_{\mathcal U}:{\mathcal U}\to Q,$
\[
 q\in Q\longrightarrow {\mathcal U}_q=\{X_0(q) + u_q/u_q\in U_q\}
\]
whose associated vector bundle is just $U$, describes the affine
nonholonomic constraints. Denote $\mathcal{U}^+$ the affine dual
bundle associated to $\mathcal{U}$ whose fiber at $q \in Q$ consists
in the affine functions over $\mathcal{U}_q$. Moreover,
$\mathcal{U}^+$ has a distinguished section $\phi:Q\to {\mathcal
U}^+$ which is induced by the constant function $\phi_q=1$ on
${\mathcal U}_q$.

On the other hand, if we denote by  $\widetilde{\mathcal
U}=({\mathcal U}^+)^*$ the bidual bundle of ${\mathcal U}$, then
$\widetilde{\mathcal U}$ is a vector subbundle of $\R\times E\to Q$
with  fiber at  $q\in Q$
\[
\widetilde{\mathcal U}_q=\{(\lambda, \lambda X_0(q) + u_q)/\lambda\in {\R} \mbox{ and } u_q\in U_q\}.
\]
Thus, a section of $\widetilde{\mathcal U}$ may be identified with a
pair $(f,fX_0+\sigma),$ with $\sigma\in \Gamma(U)$ and $f\in
C^\infty(Q).$ Under these identifications, the distinguished section
$\phi$ is given by
\[
\phi(f,fX_0+\sigma)=f.
\]
Moreover, in a natural way, the projector ${\mathcal P}:E=U\oplus
U^\perp \to U$ defined by the metric ${\mathcal G}$  induces  a new
morphism $\widetilde{\mathcal P}:\R\times E\to \widetilde{\mathcal
U}$ of vector bundles given by
\[
\widetilde{{\mathcal P}}(\lambda,e_q)=(\lambda, \lambda X_0 + {\mathcal P}(e_q)),
\]
for all $\lambda\in \R$ and $e_q\in E_q.$

In what follows, we will  introduce a skew-symmetric algebroid
structure on $\widetilde{\mathcal U}$ such that $\phi$ is a
$1$-cocycle. In order to do this, we consider the Lie algebroid
structure $(\lcf\cdot,\cdot\rcf_{\R\times E},\rho_{\R\times E})$ on
$\R\times E$ induced by the Lie algebroid structure on $E$ and the
homomorphism $F\equiv0$ on $E$ (see Example \ref{ejemplo}). Now, we
consider  the bracket $\lcf\cdot,\cdot\rcf_{\widetilde{\mathcal U}}$
on the space of sections of ${\widetilde{\mathcal U}}$ and the
vector bundle morphism $\rho_{\widetilde{\mathcal
U}}:{\widetilde{\mathcal U}}\to TQ$ given by
\[
\lcf (f_1,f_1X_0+\sigma_1),(f_2,f_2X_0+\sigma_2)\rcf_{\widetilde{\mathcal U}}=\widetilde{\mathcal P}( \lcf (f_1,f_1X_0+\sigma_1),(f_2, f_2X_0+\sigma_2)\rcf_{\R\times E})
\]
\[\rho_{\widetilde{\mathcal U}}(f,fX_0+\sigma)=\rho_{\R\times E}(f,fX_0+\sigma)=\rho(fX_0+\sigma)
\]
for $\sigma,\sigma_1,\sigma_2\in \Gamma(U)$ and $f,f_1,f_2,f\in C^\infty(Q).$ Then, using that ${\mathcal P}:\R\times E\to \widetilde{U}$ is a projector, we deduce that the pair $(\lcf\cdot,\cdot\rcf_{\widetilde{\mathcal U}},\rho_{\widetilde{\mathcal U}})$ is a skew-symmetric algebroid structure on ${\widetilde{\mathcal U}}.$  With respect to this structure,  one may prove  that \[
d^{\widetilde{\mathcal U}}\phi=0.
\]

Note that the corresponding skew-symmetric algebroid structure on $\widehat{\phi}^{-1}(0)\cong U$ is just
\[
\lcf \sigma_1,\sigma_2\rcf_U={\mathcal P}(\lcf
\sigma_1,\sigma_2\rcf),\;\;\;\,
\rho_U(\sigma)=\rho(\sigma),\;\;\;\mbox{ with }\sigma_i,\sigma\in
\Gamma(U).
\]
Moreover, ${\mathcal P}:E\to U$ and $\widetilde{\mathcal P}: \R\times E\to {\widetilde{\mathcal U}}$ are skew-symmetric algebroid morphisms.

On the other hand, one may consider  the hamiltonian section $h:U^*\to {{\mathcal U}}^+$ defined by
\[
h(\eta_q)(\lambda, \lambda X_0(q)+u_q)=\eta_q(u_q)-\lambda H(\eta_q),\;\;\;\; \forall \eta_q\in U^*_q, \mbox{ and } (\lambda,\lambda X_0(q) + u_q)\in \widetilde{\mathcal U}_q,
\]
where $H:U^*\to \R$ is the real function
\[
H(\eta_q)=\frac{1}{2}{\mathcal G}_{U^*}(\eta_q,\eta_q) + {\mathbb V}(q).
\]
Here,  ${\mathcal G}_{U^*}:U^*\times U^*\to \R$ is  the fiber metric induced by ${\mathcal G}$ on $U^*.$ In this case, we have that
$F_h=\widehat{(1,X_0)} + H\circ \mu, $ where $\widehat{(1,X_0)}$ is the linear function on ${\mathcal U}^+$ induced by the section $(1,X_0)\in \Gamma(\widetilde{\mathcal U}).$

Let $(q^i)$ be a system of local coordinates for $Q$ and consider an
orthonormal local basis $\{e_a,e_A\}$ of $\Gamma(E)$ adapted to the decomposition $E=U\oplus U^\perp$. Then, $\{(1,X_0), (0,e_a)\}$ is a local basis of sections of $\widetilde{\mathcal U}$. Denote by $(q^i,p_0,p_a)$  (respectively, $(q^i,p_a)$) the corresponding local coordinates
on ${\mathcal U}^+$ (respectively, $U^*$) with respect to the dual basis of $\{(1,X_0), (0,e_a)\}$ (respectively $\{e_a\}$).

The local expression of the hamiltonian section  $h\in
\Gamma(\mu)$ is the following
\[
\begin{array}{rcl}
h(q^i,p_a)&=&(q^i,-H=\displaystyle\displaystyle-\frac{1}{2}(p_a)^2-
{\mathbb V}(q), p_a).
\end{array}
\]
The integral curves of the hamiltonian vector field $R_h\in
{\mathfrak X}(U^*)$ are the solutions of the Hamilton equations
\[
\begin{array}{rcl}
\displaystyle\frac{dq^i}{dt}&=& \rho_0^i + \rho_a^ip_a\\[8pt]
\displaystyle\frac{dp_b}{dt}&=&-\rho_b^i\displaystyle\frac{\partial
{\mathbb V} }{\partial q^i} + (\C_{0b}^c + \C_{ab}^cp_a)p_c,
\end{array}
\]
where ${\mathcal P}(\lcf X_0,e_b\rcf)={\mathcal C}_{0b}^ce_c,$ $\rho(X_0)=\rho_0^i\frac{\partial}{\partial q^i}$ and ${\mathcal C}_{ab}^c$ and $\rho^i_a$ are local structure functions of $E.$ A Lagrangian version of these equations was considered in \cite{IMMS}.

Now, let $\alpha$ be  a section  of $U^*.$ In this case, we have that the section of  ${\widetilde{\mathcal U}}$ defined as in (\ref{3.2}) is just
\[
\zeta_h^\alpha=(1,\zeta_H^\alpha +  X_0),
\]
where  $\zeta^\alpha_H$  is the section of $U$ given by
$$\eta(\zeta_H^\alpha)=\eta^\vee(H)\circ \alpha, \;\;\; \forall \eta\in \Gamma(U^*).$$

Therefore, from Theorem \ref{main}, we deduce that

\begin{corollary} \label{Cor:Affine}
For $\alpha \in \Gamma(U^*)$, the following statements are
equivalent:
\begin{enumerate}
\item[$(i)$] If $c:I\to Q$ is an integral curve of $R_{h}^\alpha=
T\tau_{U^*}\circ R_{h}\circ \alpha\in {\mathfrak
X}(Q)$,  then  $ \alpha\circ c:I\to U^*$ is a solution of the Hamilton
equations.
\item[$(ii)$] $\alpha\in \Gamma(U^*)$ satisfies the {
Hamilton-Jacobi equation}:
$$i_{\zeta_H^\alpha} d^U\alpha+\mu\circ i_{(1,X_0)}d^{\widetilde{\mathcal U}}(h\circ \alpha)=0.$$
\end{enumerate}
\end{corollary}

\bigskip

\begin{remark}{\rm
The section $\omega^\alpha_h=\mu\circ
i_{(1,X_0)}d^{\widetilde{\mathcal U}}(h\circ \alpha)$ on $U^*$ can
be written as
\[
\omega_h^\alpha(X)=\rho(X_0)(\alpha(X))+\rho(X)(H\circ \alpha)-\alpha({\mathcal P}(\lcf X_0,X\rcf)).
\]
}
\end{remark}

From Corollary \ref{c2}. we conclude that
\begin{corollary} \label{Cor:Affine2}  Suppose that $Q$ is a connected manifold and that the finitely
 generated distribution ${\mathcal V}$ defined by ${\mathcal V}_q:=\rho_U(U_q)$ for all $q\in Q$,
is a completely nonholonomic distribution.  If  $\beta$ is a section
of ${\mathcal U}^+$ such that $d^{\widetilde{\mathcal U}}\beta=0$,
 then the  following statements are equivalent:
\begin{enumerate}
\item[$(i)$] If $c:I\to Q$ is an integral curve of $R_{h}^{\mu\circ \beta}=
T\tau_{U^*}\circ R_{h}\circ \mu\circ \beta\in {\mathfrak
X}(Q)$,  then  $ \mu\circ \beta\circ c:I\to V^*$ is a solution of
the Hamilton equations.
\item[$(ii)$] $\beta\in \Gamma({\mathcal U}^+)$ satisfies the {
Hamilton-Jacobi equation}:
\[
H\circ \mu\circ \beta + \beta(1,X_0)=constant.
\]

\end{enumerate}
\end{corollary}

\begin{example}\label{ex:ball}{\emph{ An homogeneous rolling ball without sliding
on a rotating table with time-dependent angular velocity. }}
{\rm We consider a homogeneous ball with radius $r>0$, mass $m$ and
inertia $mk^2$ about any axis. Suppose that the ball rolls without
sliding on a horizontal table which rotates with a time-dependent
angular velocity $\Omega(t)$ about vertical axis through of one of
its points. Apart from the gravitational force, no other external
forces are assumed.

Choose a cartesian reference frame with origin at the center of
rotation of the table and $z-$axis along the rotation axis. If $(t,q^1,q^2,\dot{q}^1,\dot{q}^2,\omega_1,\omega_2,\omega_3)$ are the corresponding coordinates over $\R\times T\R^2\times \R^3,$ then $(q^1,q^2)$ denotes the position of the point of contact of the
sphere with the table and  $\omega_1, \omega_2$ and $\omega_3$ are the
components of the angular velocity of the sphere.

Note that since the ball is rolling without sliding, then the system
is subjected to the affine constraints
 \begin{eqnarray*}
\dot q^1 - r \omega_2 &=&
- \Omega(t) q^2 \\
\dot q^2 + r \omega_1 &=& \Omega(t) q^1.
\end{eqnarray*}

Let $(t,q^1,q^2,p_1,p_2,\pi_1,\pi_2,\pi_3)$ be the corresponding coordinates on $(\R\times T\R^2\times \R^3)^*.$ The hamiltonian section $h:(\R\times T\R^2\times \R^3)^*\to \R\times (\R\times T\R^2\times \R^3)^*$  of the system is given by
$$h(t,q^1,q^2,p_1,p_2,\pi_1,\pi_2,\pi_3)=(-H(t,q^1,q^2,p_1,p_2,\pi_1,\pi_2,\pi_3), t,q^1,q^2,p_1,p_2,\pi_1,\pi_2,\pi_3)$$
where $H:(\R\times T\R^2\times \R^3)^*\to \R$ is the real function
$$H=\frac{1}{2} (\frac{1}{m}(p_1^2 + p_2^2) + \frac{1}{mk^2}(\pi_1^2 + \pi_2^2 + \pi_3^2)).$$
Moreover, the constraints may be rewritten as follows
\begin{eqnarray*}
\psi_1&=&\Omega(t)q^2 + \frac{1}{m}p_1 -\frac{r}{mk^2}\pi_2=0 \\
\psi_2&=&-\Omega(t)q^1 + \frac{1}{m}p_2 +\frac{r}{mk^2}\pi_1=0.
\end{eqnarray*}

Then the Hamilton equations of this non-holonomic system are
$$\begin{array}{rcl}
\dot{q^1}&=&\displaystyle\frac{1}{m}p_1\\[8pt]
\dot{q^2}&=&\displaystyle\frac{1}{m}p_2\\[8pt]
\dot{p_1}&=& -\displaystyle\frac{mk^2}{k^2 + r^2} (\frac{d\Omega(t)}{dt}q^2 + \Omega(t) \frac{p_2}{m})\\[8pt]
\dot{p_2}&=& \displaystyle\frac{mk^2}{k^2 + r^2} (\frac{d\Omega(t)}{dt}q^1 + \Omega(t) \frac{p_1}{m})\\[8pt]
\dot{\pi_1}&=& \displaystyle\frac{rmk^2}{k^2 + r^2} (\frac{d\Omega(t)}{dt}q^1 + \Omega(t) \frac{p_1}{m})\\[8pt]
\dot{\pi_2}&=&\displaystyle\frac{rmk^2}{k^2 + r^2} (\frac{d\Omega(t)}{dt}q^2 + \Omega(t) \frac{p_2}{m})\\[8pt]
\dot{p_3}&=&0
\end{array}$$
and the constraints $\psi_1=\psi_2=0$ (for more details, see \cite{IMMS}; see also \cite{DianaT}).

Our goal is to encode all this information in a mechanical system subjected to affine nonholonomic constraints on a Lie algebroid. Consider the vector bundle $\tau_E:E\to
Q,$ where $E:= T\R ^3\times \R^3$, $Q=\R^3$ and $\tau_E:E\to Q$ is  defined by
\[
\tau_E(t,q^1,q^2,\dot{t},
\dot{q}^1,\dot{q}^2,\omega_1,\omega_2,\omega_3)=(t, q^1,q^2).
\]
We choose the
following global basis of $\Gamma(E)$
\[
\begin{array}{lll}e_0=(\displaystyle\frac{\partial }{\partial t}-\Omega(t)q^2\displaystyle\frac{\partial }{\partial q^1} + \Omega(t)q^1\displaystyle\frac{\partial }{\partial q^2},0)&e_1=(\displaystyle\frac{\partial }{\partial q^1},0),&
e_2=(\displaystyle\frac{\partial }{\partial q^2},0),\\[8pt]
e_3=(0,(1,0,0)),& e_4=(0,(0,1,0)),&e_5=(0,(0,0,1)),\end{array}\]
On $E$ we define the Lie algebroid structure by
\[
\lcf e_0,e_1\rcf=-\Omega(t)e_2,\;\;\; \lcf e_0,e_2\rcf=\Omega(t)e_1,\;\;\;
\lcf e_3,e_4\rcf_E=e_5,\;\;\; \lcf e_4,e_5\rcf_E=e_3,\;\;\;\lcf
e_5,e_3\rcf_E = e_4,\]
\[\rho_E(e_0)=\displaystyle\frac{\partial }{\partial t}-\Omega(t)q^2\displaystyle\frac{\partial }{\partial q^1} + \Omega(t)q^1\displaystyle\frac{\partial }{\partial q^2},\;\;\;\;\rho_E(e_1)=\displaystyle\frac{\partial }{\partial q^1},\;\;\;\; \rho_E(e_2)=\displaystyle\frac{\partial }{\partial q^2}.\]
The rest of the local structure functions are zero.

The constraints induce a vector subbundle of $E$
$$U:= span \{ e_3-re_2\, , \, e_4+re_1\, , \, e_5 \}.$$

Consider  the bundle metric on $E$
$${\mathcal G}={e_0}^2 + (m((e_1)^2 + (e_2)^2) + mk^2((e_3)^2 + (e_4)^2 + (e_5)^2).$$

In order to do the decomposition $E=U\oplus U^\perp,$ we take  the following orthonormal basis of $\Gamma(E)$ with
respect to ${\mathcal G}$
\[
\begin{array}{llll}
\bar{e_0}=e_0,&
\bar{e}_1=\displaystyle \frac{1}{k\sqrt{m(k^2+r^2)}} ({r}e_3
+ k^2e_2),&
\bar{e}_2=\displaystyle \frac{1}{k\sqrt{m(k^2+r^2)}}({r}e_4-k^2e_1),\\[8pt]
\bar{e}_3=\displaystyle\frac{1}{\sqrt{m(k^2+r^2)}}(e_3-re_2),&
\bar{e}_4=\displaystyle\frac{1}{\sqrt{m(k^2+r^2)}}(e_4 + re_1),&
\bar{e}_5=\displaystyle\frac{1}{k\sqrt{m}}e_5.\end{array}\]
Then, $\{\bar e_3, \bar e_4, \bar e_5 \}$ (respectively, $\{\bar{e_0}, \bar e_1,\bar e_2\}$) is a orthonormal basis of $U$ (respectively, $U^\perp$).

Moreover, for this mechanical system, the distinguished section $X_0$ of $E$ is
$X_0=\bar{e}_0.$
Note that ${\mathcal P}(X_0)=0.$

Denote by $(t,q^1,q^2,\bar p_0, \bar p_1,\bar p_2,\bar \pi_1,\bar \pi_2,\bar \pi_3)$ the coordinates on $E^*$ with respect to the dual basis $\{\bar e^0, \bar e^1, \bar e^2, $ $\bar e^3, \bar e^4, \bar e^5 \}$ of $\{\bar e_0, \bar e_1, \bar e_2, \bar e_3, \bar e_4, \bar e_5 \}$. With respect to these coordinates the function $H:U^*\to \R$ is defined by
$$H(\bar \pi_1,\bar \pi_2,\bar \pi_3)=\frac{1}{2}(\bar \pi_1^2 + \bar \pi_2^2 +\bar \pi_3^2)$$
and the structure functions of the skew-symmetric algebroid on
$\widetilde{\mathcal U}\to \R^3$ with respect to the basis
$\{(1,X_0),(0,\bar e_i)\}_{i=3,4,5}$ are the following
$$\begin{array}{c}
\bar{\C}_{34}^5=\bar{\C}_{45}^3=\bar{\C}_{53}^4=\displaystyle\frac{k}{\sqrt{m}(r^2+ k^2)},\\[8pt] \bar{\rho}_0^0=1,\;\;\;
\bar{\rho}_0^1=-\Omega(t)q^2, \;\;\;\bar{\rho}_0^2=\Omega(t)q^1,\;\;\;
\bar{\rho}_4^1=-\bar{\rho}_3^2=\displaystyle\frac{r}{\sqrt{m(r^2+
k^2)}}.
\end{array} $$

Let us consider the section $\alpha \in \Gamma(U^*)$ to be $\alpha=
d^Ug$ for the real function on $\R^3$
$$g=g(t,q^1,q^2)= \varphi_1 (t)
q^1 + \varphi_2(t)q^2$$
 where $\varphi_1, \varphi_2 \in
C^\infty(\R)$. Then, $$\alpha =
\frac{r}{\sqrt{m(k^2+r^2)}}(- \frac{\partial g}{\partial q^2} \bar e
^3 +  \frac{\partial g}{\partial q^1}
\bar e ^4)=
\displaystyle\frac{r}{\sqrt{m(k^2+r^2)}} (-\varphi_2(t) \bar e ^3 +
 \varphi_1(t) \bar e ^4) $$
and the section $\zeta_H^\alpha\in\Gamma(U)$ is
$$ \zeta_H^\alpha=
\displaystyle\frac{r}{\sqrt{m(k^2+r^2)}} (-\varphi_2(t) \bar e _3 +
 \varphi_1(t) \bar e _4).$$

It is important to note that $\alpha \in \Gamma(U^*)$ is not a
1-cocycle of the skew-symmetric algebroid $\tau_{U^*}:U^*
\rightarrow Q$. In fact, $$d^U\alpha = d^U(d^U g)= \frac{k r}{m(k^2+
r^2)^{3/2}}(\varphi_1(t) {\bar e}^3 + \varphi_2(t) {\bar e}^4)\wedge
{\bar e}^5 \not= 0.$$

However,
\[
i_{\zeta_H^\alpha}d^U\alpha=0.
\]
Thus,  the Hamilton-Jacobi equation becomes
\begin{equation}\label{HJ1}\dot \varphi_2(t)= \frac{\Omega(t) r^2}{k^2+r^2} \varphi_1(t) \qquad
\dot \varphi_1(t)= -\frac{\Omega(t) r^2}{k^2+r^2} \varphi_2(t).\end{equation}

Now, in order to apply Corollary \ref{Cor:Affine}, we have to find
an integral curve $c(s)=(t(s),q^1(s),q^2(s))$, for $s \in \R$ of the
vector field $R_h^\alpha \in \frak{X}(Q)$ given by
$$R_h^\alpha =
\frac{\partial}{\partial t}+ \left( \frac{r^2}{m(k^2+r^2)} \varphi_1
- \Omega(t) q^2 \right) \frac{\partial}{\partial q^1} + \left(
\frac{r^2}{m(k^2+r^2)} \varphi_2 + \Omega(t) q^1 \right)
\frac{\partial}{\partial q^2}.$$ Then, the curve $c$ has to verify
$t(s)=s+c_0$, and taking $c_0=0$ we get
\begin{equation} \begin{array}{rcl} \dot q^1(t) &=& \displaystyle \frac{r^2}{m(k^2+r^2)} \varphi_1(t) - \Omega(t)
q^2(t)\\[8pt]
\dot q^2(t) &=& \displaystyle \frac{r^2}{m(k^2+r^2)} \varphi_2(t)
+ \Omega(t) q^1(t). \label{ExAff:system} \end{array} \end{equation}
We conclude that $\displaystyle{ \alpha \circ c(t) = (t,
q^1(t),q^2(t) ; \frac{-r}{\sqrt{m(k^2+r^2)}} \varphi_2(t) ,
\frac{r}{\sqrt{m(k^2+r^2)}} \varphi_1(t),0)}$ is an integral curve
of $R_h$, where $\varphi_i(t)$ and $q^i(t)$ are real functions that satisfy (\ref{HJ1}) and (\ref{ExAff:system}).

As a particular case, we can take the angular velocity of the table
to be $\Omega(t)=\Omega_0=cte >0$ and we get that the curve $\alpha
\circ c(t) = (t, q^1(t),q^2(t) ; \alpha_3(c(t)),\alpha_4(c(t)),0)$
is given by
\begin{eqnarray*} \alpha_3(c(t)) &=& \frac{-r}{\sqrt{m(k^2+r^2)}} \left( C_1 \sin\left(\frac{r^2 \Omega_0
t}{k^2+r^2}\right)+
C_2 \cos\left(\frac{r^2 \Omega_0 t}{k^2+r^2}\right)\right),\\
\alpha_4(c(t))&=& \frac{r}{\sqrt{m(k^2+r^2)}} \left(C_1
\cos\left(\frac{r^2 \Omega_0 t}{k^2+r^2}\right) - C_2
\sin\left(\frac{r^2 \Omega_0 t}{k^2+r^2}\right)\right),
\end{eqnarray*} and $q^1(t), q^2(t)$ solutions of the system
(\ref{ExAff:system}). The trajectories of the ball on the rotating
table (trajectories in $(q_1(t),q_2(t))$) are ellipses centered in
the origin of the table, which depend on the initial conditions of
the problem.

If $\Omega(t)= \Omega_0 t$ then the curve $\alpha \circ c(t) = (t,
q^1(t),q^2(t) ; \alpha_3(c(t)),\alpha_4(c(t)),0)$ is given by
\begin{eqnarray*}
\alpha_3(c(t)) &=&  \frac{-r}{\sqrt{m(k^2+r^2)}} \left( C_1
\sin\left(\frac{r^2 \Omega_0 t^2}{2(k^2+r^2)}\right)
 + C_2 \cos\left(\frac{r^2 \Omega_0 t^2}{2(k^2+r^2)}\right) \right) ,\\ \alpha_4(c(t))
 &=& \frac{r}{\sqrt{m(k^2+r^2)}} \left(C_1 \cos\left(\frac{r^2 \Omega_0 t^2}{2(k^2+r^2)}\right)
  - C_2 \sin(\frac{r^2 \Omega_0 t^2}{2(k^2+r^2)}) \right), \end{eqnarray*}
where $C_1,C_2$ are real constants.  In this case, the solutions
$(q_1(t),q_2(t))$ of
  (\ref{ExAff:system}) give trajectories on the table as in Figure 2.
\begin{figure}[h]
\includegraphics[width=14cm]{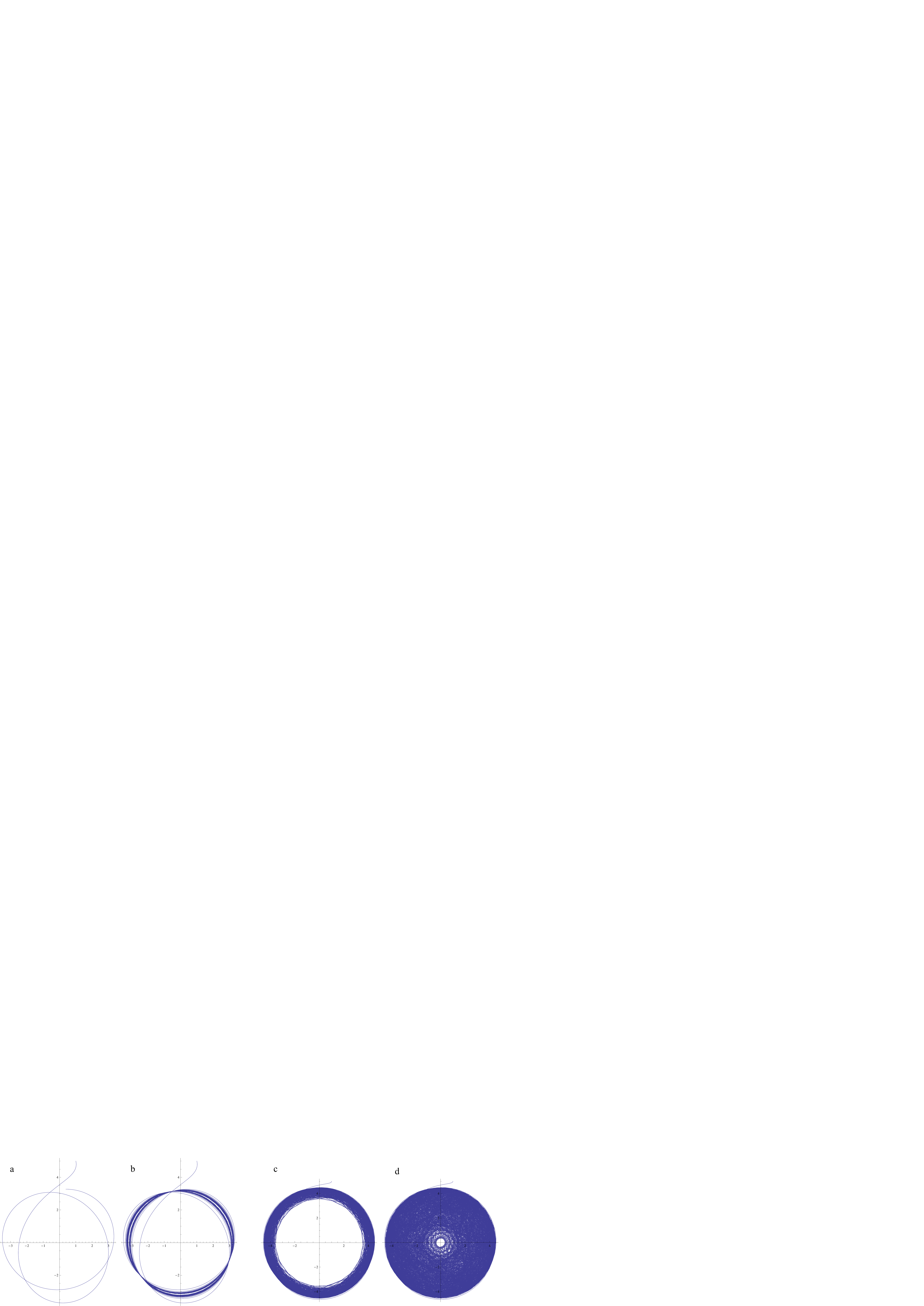}
 \caption{The trajectory of the ball on the plane with velocity $\Omega(t)=t$ for $0\leq t \leq 5$ and $0 \leq t \leq 30.$
 In the last two figures the velocity is changed
 $\Omega(t)=10 t$ for $0 \leq t\leq 20$ and $0 \leq t\leq 30$}
\end{figure}
}
\end{example}


\subsection{ A example of a nonholonomic mechanical system with linear external forces:
 the vertical rolling disk with external forces}\label{ex:disc}
 We will use the classical example of the vertical rolling disk
to show how external forces can be encoded in the geometric
structure of the constraint submanifold. Then, we are going to find
the Hamilton-Jacobi equation and we obtain some particular solutions.

Consider a vertical disk that is allowed to roll on the $xy$-plane
and to rotate about its vertical axis. Let $x,y$ denote the position
of contact of the disk with the $xy$-plane, $\theta$ will denote the
rotation angle of a chosen point $P$ of the disk with respect to the
vertical axis and finally $\phi$ will represent the orientation
angle of the disk as in Figure 3.
\begin{figure}[h]
 \includegraphics[height=4cm, width=9cm]{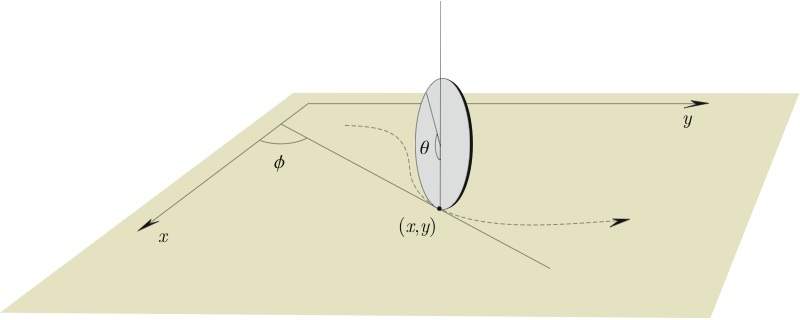}\\
 \caption{The vertical rolling disk}
\end{figure}

Therefore, the configuration space for the rolling disk is $Q=\R^2
\times S^1\times S^1$ with coordinates $(x,y,\theta,\phi)$. On the
tangent bundle $TQ \rightarrow Q$ we consider the Lie algebroid
structure $([\cdot , \cdot ]_{TQ}, id_{TQ})$ where $[\cdot , \cdot
]_{TQ}$ is the usual Lie bracket of vector fields and the anchor
map, in this case, is $id_{TQ}$.

The Lagrangian for this system is:
\[
L(x,y,\theta,\phi; \dot x, \dot y, \dot \theta, \dot
\phi)=\frac{1}{2} \left( m(\dot{x}^2+\dot{y}^2) + I \dot \theta^2 +
J \dot \phi^2 \right)
\]
where $m$ is the mass of the disk, $I$ its  moment of inertia about
the axis perpendicular to the plane containing the disk  and $J$ is
the moment of inertia about an axis in the plane of the disk. This
Lagrangian induces a fiber metric
$$\Gc= m (dx \otimes dx + dy \otimes dy) + I d\theta \otimes d\theta + J d\phi \otimes d\phi.$$

The nonholonomic constraints of rolling without slipping are
$$\left\{ \begin{array}{l} \dot x = (R \cos \phi) \dot \theta\\
\dot y = (R\sin \phi )\dot \theta \end{array}\right.$$ and they
define the constraint subbundle $\tau_D : D \rightarrow Q$ of $TQ$.

In terms of the fiber metric $\Gc$, we find an adapted basis for the
nonholonomic problem. More precisely, we look for an orthonormal
basis of vector fields $\{X_1,X_2, X_3,X_4\}$ of $TQ$
such that $D = span \{X_1,X_2\}$ and $D^{\perp} = span\{X_3,X_4\}$.
This basis is given by
\begin{eqnarray*} X_1 &=& \frac{1}{\sqrt{R^2 m + I}}
\left( R \cos \phi \frac{\partial}{\partial x} + R \sin \phi \frac{\partial}{\partial y} +
\frac{\partial}{\partial \theta} \right) \\
X_2&=& \frac{1}{\sqrt{J}} \frac{\partial}{\partial \phi} \\
X_3&=& \frac{1}{\sqrt{m}} \left( \sin \phi \frac{\partial}{\partial x} - \cos \phi \frac{\partial}{\partial y} \right)\\
X_4 &=& \sqrt{\frac{I}{m(R^2 m + I)}} \left( \cos \phi
\frac{\partial}{\partial x} + \sin \phi \frac{\partial}{\partial y}
- \frac{Rm}{I} \frac{\partial}{\partial \theta} \right).
\end{eqnarray*}

We endow the fiber bundle $\tau_D: D \rightarrow Q$ with a
skew-symmetric algebroid structure $(\lcf \cdot, \cdot \rcf _D,
\rho_D)$ defined by (see Example \ref{projector})
$$\lcf X_1, X_2 \rcf _D = P([X_1, X_2 ]_{TQ}) \quad \mbox{and} \quad \rho_D(X_1) = X_1, \ \rho_D(X_2) = X_2,$$
where $P:TQ \rightarrow D$ is the orthogonal projector (with respect to the decomposition
$TQ=D \oplus D^{\perp}$). Note that, $\rho_D=\rho_{TQ} \circ i_D$ with $i_D:D \hookrightarrow TQ$
the natural inclusion. Therefore, in terms of the basis $\{X_1, X_2\}$, the (non zero) local structure
functions of the skew-symmetric algebroid on $D$ are given by
\begin{eqnarray} (\rho_D)_1^x = \frac{R \cos \phi}{\sqrt{m R^2 + I}}, & \ & \ \  (\rho_D)_1^\theta =
 \frac{1}{\sqrt{m R^2 + I}}, \nonumber \\
 (\rho_D)_1^y = \frac{R \sin \phi}{\sqrt{m R^2 + I}}, & \ & \ \ (\rho_D)_2^\phi = \frac{1}{\sqrt{J}}.
 \label{Ex:D StrucFunc} \end{eqnarray}
Since $[X_1,X_2]_{TQ} \in span \{X_3\}$ we have that $\C_{12}^1=\C_{12}^2=0$.

In coordinates \ $(v^1, v^2,v^3,v^4)$ \ induced  by the orthonormal basis of sections
  \ $\{X_1,X_2,X_3,X_4\}$ \ the Lagrangian is $$L(x,y,\theta,\phi; v^1, v^2,v^3,v^4) =
  \frac{1}{2} \left( (v^1)^2+ (v^2)^2 + (v^3)^2+(v^4)^2\right),$$ and the equations determining the constraints
   are $v^3 = v^4=0$. Therefore, the restricted lagrangian $L_D:D \rightarrow \R$ becomes
    $L_D(x,y,\theta,\phi; v^1, v^2) = \frac{1}{2} \left((v^1)^2+ (v^2)^2 \right)$.

Consider now, the dual vector bundle \ $\tau_{D^*}: D^* \rightarrow
Q$ \ \  with coordinates \ $(x,y,\theta,\phi; p_1, p_2)$ induced by
the dual basis $\{X^1, X^2\}$ of $\{X_1, X_2\}$. Then, the vector
bundle $\tau_{D^*}: D^* \rightarrow Q$ has a linear almost Poisson
structure given by
\begin{eqnarray*} \{x,p_1\}_{D^*} = \frac{R \cos \phi}{\sqrt{m R^2 + I}} ,  & \ & \ \ \{\theta,p_1\}_{D^*} =
 \frac{1}{\sqrt{m R^2 + I}}, \\
\{y,p_1\}_{D^*} = \frac{R \sin \phi}{\sqrt{m R^2 + I}} , & \ & \ \
\{\phi,p_2\}_{D^*} = \frac{1}{\sqrt{J}}
\end{eqnarray*} and the other fundamental  brackets are zero.

In these coordinates, the Hamiltonian function $H:D^*\rightarrow \R$
can be written as
$$H(x,y,\theta,\phi; p_1, p_2)= \frac{1}{2} \left((p_1)^2+(p_1)^2\right).$$

It is very interesting the study of the rolling disk with external
forces \cite{lewis}. The system has two natural inputs, a torque
that makes the disk spin and another one that makes the disk roll.
First we are going to study the most general situation  and then we
will analyze particular cases. Suppose that a linear force is acting
on the disk, then the pull back of this force in $D^*$ is given by
$\tilde F(q,v)= (\tilde F_1^1(q) v^1 + \tilde F_2^1(q) v^2)X^1(q) +
(\tilde F^2_1(q) v^1 + \tilde F_2^2(q) v^2)X^2(q),$ where
$(q,v)=(x,y,\theta,\phi; v^1, v^2)$  and $F_i^j \in C^\infty(Q)$.

Since the chosen basis is orthonormal, we have that the homomorphism
$F:D \rightarrow D$ induced by the force $\tilde F$ is
\begin{eqnarray*} F(X_1) & =& \tilde F_1^1 X_1 + \tilde F_1^2 X_2 \\
F(X_2) &=& \tilde F_2^1  X_1 + \tilde F_2^2 X_2 \end{eqnarray*} and
thus the skew-symmetric algebroid on $\R \times D$ has (non zero)
local structure functions given by $\C_{01}^1= -\tilde F_1^1,
\C_{01}^2= -\tilde F_1^2, \C_{02}^1 = -\tilde F_2^1, \C_{02}^2 =
-\tilde F_2^2$ and equation (\ref{Ex:D StrucFunc}).

Therefore, the corresponding Hamilton equations modified by the
action of an external force are
\begin{eqnarray*} \dot x = \frac{R \cos\phi}{\sqrt{I + m R^2}} p_1,
 \ & \ & \ \ \dot y = \frac{R\sin\phi }{\sqrt{I + m R^2}} p_1,
  \\ \dot \theta =\frac{1}{\sqrt{I + m R^2}} p_1,
   \ & \ & \ \ \dot \phi = \frac{1}{\sqrt{J}}p_2,
    \\ \dot p_1 = - \tilde F_1^1 p_1 - \tilde F_1^2 p_2
     \ & \ & \ \ \dot p_2 = - \tilde F_2^1 p_1 - \tilde F_2^2 p_2.
\end{eqnarray*}

In order to write the Hamilton-Jacobi equations, let us consider a
section $\alpha \in \Gamma(D^*)$.

Then, such equations are
\begin{equation} \begin{array}{rcl} \alpha_1.X_1(\alpha_1) + \alpha_2.X_1(\alpha_2) + \alpha_2.X_2(\alpha_1)
-\alpha_2.X_1(\alpha_2) + \alpha_1.\tilde F_1^1 + \alpha_2.\tilde F_1^2 &=&0 \\
\alpha_1 . X_2(\alpha_1) + \alpha_2. X_2(\alpha_2) - \alpha_1 .
X_2(\alpha_1) + \alpha_1. X_1(\alpha_2) + \alpha_1. \tilde F_2^1 +
\alpha_2.\tilde F_2^2 &=&0
\end{array}  \label{Ex:D HJeq} \end{equation}
where $\alpha = \alpha_1 X^1 + \alpha_2 X^2$ and $\alpha_1 ,
\alpha_2 \in C^\infty(Q)$.

 {\bf Particular case: A torque that makes the disk spin.}

Let us consider the external force $\tilde F = \lambda(\phi) \dot
\phi\, d\phi,$ with $\lambda \in C^{\infty}(\R)$. Writing this force
in terms of the dual basis $\{X^1,X^2\}$ we obtain $$\tilde F(q,v) =
\frac{\lambda(\phi)}{{J}} v^2  X^2,$$ where $(q,v)=(x,y,\theta,\phi,
v^1,v^2).$ Therefore, the homomorphism $F:D \rightarrow D$ is
$$F(X_1) = 0 \quad \mbox{and} \quad F(X_2) =
\frac{\lambda(\phi)}{{J}} X_2$$ and the skew-symmetric algebroid on
$\R \times D$ has (non zero) local structure functions given by
$\C_{02}^2 =- \frac{\lambda(\phi)}{J}$ and (\ref{Ex:D StrucFunc}).

Consider a section $\alpha \in
\Gamma(D^*)$ such that $\alpha = k X^1 +
\alpha_2(\phi) X^2$, with $k = constant.$

Thus,
Hamilton-Jacobi equation (\ref{Ex:D HJeq}) is simply (note that, in
this case, $d^E \alpha=0$),
\begin{equation} \alpha_2'(\phi) = - \frac{\lambda(\phi)}{\sqrt{J}}. \label{Ex:D HJeq2} \end{equation}
Therefore, from (\ref{Ex:D HJeq2}), we deduce that
$$\alpha_2(\phi)= - \frac{1}{\sqrt{J}}\int_0^\phi \lambda(s)ds + \kappa $$
where $\kappa$ is an arbitrary constant.

By Eq. (\ref{Rha}), we have $$R_{h}^\alpha = \frac{R\, k
\cos\phi}{\sqrt{I + m R^2}} \frac{\partial}{\partial x} + \frac{R\,
k \sin\phi}{\sqrt{I + m R^2}} \frac{\partial}{\partial y} + \frac{
k}{\sqrt{I + m R^2}} \frac{\partial}{\partial \theta} -
\frac{1}{{J}} \left(\int_0^\phi \lambda(s) ds - \kappa \right)
\frac{\partial}{\partial \phi}.$$ We conclude, by Corollary
\ref{c3}, that $$\alpha \circ c(t) = (x(t),y(t), \theta(t),\phi(t);
k , - \frac{1}{\sqrt{J}}\int_0^{\phi(t)} \lambda(s) ds -\kappa)$$
is an integral curve of $R_{h}\in {\mathfrak X}(D^*)$, if
$c(t)=(x(t),y(t), \theta(t),\phi(t))$ is an integral curve of
$R_{h}^\alpha$.

As a particular case, we fix $\lambda(\phi)=K \cos\phi$ with $K=cte
\neq 0$. Hence by equation (\ref{Ex:D HJeq2}) we have that
$\alpha_2(\phi) = -\frac{K}{\sqrt{J}} \sin \phi + \kappa$ but, just
for simplicity, we will choose $\kappa =0$.  If $c:I\to Q$,
$c(t)=(x(t),y(t),z(t),\theta(t),\phi(t))$,  is an integral curve of
$R_{h}^\alpha$ then $\dot \phi(t) = -\displaystyle\frac{K}{{J}} \sin
\phi$. That is, $$\phi(t) = 2 \arctan\left(e^{-\frac{K}{J}t
+\phi_0}\right)$$ with $\phi_0$ an arbitrary constant. Therefore,
the solution of the system, modified by an external force $\tilde F
=(K\cos\phi) \dot \phi \, d\phi$ that makes the disk spin, is
$$\alpha \circ c(t) = (x(t),y(t), \theta(t),\phi(t); k ,
-\frac{K}{\sqrt{J}} \sin \phi(t) ),$$  where
$x(t),y(t),\theta(t),\phi(t)$ are curves given by \begin{eqnarray*}
x (t) &=& \frac{R\, k }{\sqrt{I + m R^2}} \left(t + \frac{J}{K} \ln
\left( 1+ e^{-2\frac{K}{J}t +
2\phi_0} \right) \right) + x_0 \\
y (t) &=& \frac{J}{K}\frac{R\, k }{\sqrt{I + m R^2}}
\phi(t) + y_0\\
\theta(t) &=&  \frac{ k t}{\sqrt{I + m R^2}} + \theta_0 \\
\phi(t) &=& 2 \arctan\left(e^{-\frac{K}{J}t +\phi_0} \right)
\end{eqnarray*}
where $x_0,y_0, \theta_0, \phi_0$ are arbitrary constants.

We also have the dissipative term for this case given by $$ \{H\circ
\mu,F_h\}=-\frac{K \cos\phi}{J} (p_2)^2.$$

\begin{remark} The function $f \in C^\infty(Q)$, given by
$$f(\phi)=\frac{-1}{2J} \left(\int_0^\phi\lambda(s)ds \right)^2= -\frac{K^2}{2J} \sin^2 \phi$$
verifies that $F^*\alpha=d^D f$. Thus,  we obtain that the Hamilton
Jacobi equation can be written as $$H\circ \alpha -\frac{K^2}{2J}
\sin^2 \phi = constant,$$ on $Q$, since $D$ is a completely
nonholonomic distribution.
\end{remark}

\section{Conclusions and future work}

A Hamilton-Jacobi equation for a great variety of mechanical systems
is derived. The type of systems considered includes mechanical
systems with dissipative forces, nonholonomic system subjected to
linear or affine constraints or, even, explicitly time-dependent
mechanical systems. With this general purpose in mind, we find that
the geometric structure of skew-symmetric algebroid has the
appropriate inclusive nature, adequate to model all these different
types of mechanical systems. Adopting this point of view we prove a
general version of the Hamilton-Jacobi equation for skew-symmetric
algebroids with a distinguished cocycle, specializing the results
for the different mechanical systems under study. Several examples
prove the utility and novelty of our results.

Of course, a lot of work must be done in future research. For
instance, in our paper a crucial assumption is made: all the
constraints are linear or affine, even the dissipative forces
considered are of a very special type (in such a way that they
induce a linear bivector  on the dual bundle). It would be
interesting to discuss the more general case in a non-linear
setting, discovering the underlying geometric structures and
deriving, if possible, a Hamilton-Jacobi equation. Moreover, in
future papers, we will study more explicit examples of applications
of our theoretical setting, analyzing  when the separation of
variables technique works and relating it with topics like
integrability. Also, our setting is ready for the introduction of
control forces and therefore for the study of controlled mechanical
systems  and, as a consequence, to address problems like kinematic
reduction, kinematic controllability, Hamilton-Jacobi-Bellman
equation in optimal control, etc.



\begin{thebibliography}{99}

\bibitem{AbMa} R. Abraham, J.E. Marsden: {\sl Foundations of
Mechanics}, Second Edition, Benjamin, New York, 1978.

\bibitem{Orbit} A.A. Agrachev, Y. Sachkov: Control theory from the
geometric view-point, {\sl Encyclopedia of Mathematical Science},
Vol {\bf 87}. Control theory and optimization. Springer-Verlag,
Berlin (2004)

\bibitem{BLMM}P. Balseiro, M. de Le\'on, J.C. Marrero, D.
Mart{\'\i}n de Diego: The ubiquity of the symplectic hamiltonian
equations in mechanics, {\sl J. Geometric Mechanics, } {\bf 1}
(2009) 1-34.


\bibitem{CLMM} F. Cantrijn, M. de Le\'on, J.C. Marrero, D. Mart{\'\i}n de Diego: On almost-Poisson structures in non-holonomic mechanics II. The time dependent framework. {\sl Nonlinearity, } {\bf 13} (2000), 1379-1409.

\bibitem{TD12} F. Cantrijn, M. de Le\'on, D. Mart{\'\i}n de Diego: On almost-Poisson structures in non-holonomic mechanics;  {\sl Nonlinearity, } {\bf 12} (1999), 721-737.

\bibitem{CaGrMaMaMuRo} J. F. Cari\~nena, X. Gracia, G. Marmo, E. Mart\'\i nez, M. C. Mu\~noz-Lecanda, N. Roman--Roy
: Geometric Hamilton-Jacobi Theory for Nonholonomic Dynamical Systems, Preprint  (2009) arXiv:0908.2453.

\bibitem{CoLeMaMaMa} J. Cort\'es, M. de Le\'on, J.C. Marrero, D. Mart\'\i n de Diego, E. Mart\'\i nez:
A survey of Lagrangian mechanics and control on Lie algebroids and
groupoids, {\sl Int. J. Geom. Methods Mod. Phys.} {\bf 3} (2006),
no. 3, 509--558.

\bibitem{CoLeMaMar} J. Cort\'es, M. de Le\'on, J.C. Marrero,
    E. Mart\'{\i}nez: Nonholonomic Lagrangian systems on Lie
    algebroids, {\sl Discrete and Continuous Dynamical Systems: Series A},
    {\bf 24} (2) (2009), 213--271.

\bibitem{CoMa} J. Cort\'es, E. Mart\'{\i}nez: Mechanical
    control systems on Lie algebroids, {\sl IMA J. Math. Control.
    Inform.} {\bf 21} (2004), 457--492.

\bibitem{GrGr} K. Grabowska,   J. Grabowski:
Variational calculus with constriants on general algebroids, {\sl J.
Phys. A: Math Theoret.,} {\bf 41} (2008), 175204.

\bibitem{GGU} K. Grabowska,   J. Grabowski, P. Urba\'nski:
{AV-differential Geometry: Poisson and Jacobi structures}, {\sl
Journal of Geometry and Physics,} {\bf 52} (2004) 398--446.

\bibitem{GGU0} K. Grabowska, J. Grabowski, P. Urba\'nski: Geometrical mechanics on
algebroids, {\sl Int. J. Geom. Methods Mod. Phys.} {\bf 3} (3)
(2006), 559575.

\bibitem{GrLeMaMa}  J. Grabowski, M. de Le\'on, J.C. Marrero, D. Mart{\'\i}n de Diego:
Nonholonomic constraints: a new viewpoint, {\sl J. Math. Phys.}
{\bf 50}  (2009),  no. 1, 013520, (17 pp)


\bibitem{GU1} J. Grabowski, P. Urba\'nski: Lie algebroids and Poisson-Nijenhuis
structures, {\sl Rep. Math. Phys. } {\bf 40} (1997), 195--208.

\bibitem{GU2}J. Grabowski, P. Urba\'nski: Algebroids  general differential
calculi on vector bundles, {\sl J. Geom. Phys. } {\bf 31} (1999),
111--141.


\bibitem{IbLeMaMa} A. Ibort, M. de Le\'on, J. C. Marrero, D. Mart{\'\i}n de Diego: Dirac brackets in constrained dynamics, {\sl Fortschr. Phys.}, {\bf 47 } (1999), 459--492


\bibitem{ILM} D. Iglesias, M. de Le\'on,  D. Mart{\'\i}n de Diego:  Towards a
Hamilton-Jacobi Theory for Nonholonomic Mechanical Systems, {\sl
J. Phys. A: Math. Theor. } {\bf 41 } (2008) 015205 (14pp).


\bibitem{IMMS} D. Iglesias, J.C. Marrero, D. Mart{\'\i}n de Diego, D. Sosa:  A
general framework for nonholonomic mechanics: Nonholonomic systems
on Lie affgebroids, {\sl J. Math. Phys.}, {\bf 48} (2007) 083513.

\bibitem{TD49} D. Iglesias, J.C. Marrero,  E. Padr\'on, D. Sosa:  Lagrangian submanifolds and dynamics on Lie affgebroids,  {\sl
Rep. Math. Phys. } {\bf 38 } (2006) 385--436.


\bibitem{KooMa} W.S. Koon, J.E. Marsden: Poisson reduction of nonholonomic mechanical systems with symmetry, {\sl
Rep. Math. Phys.} {\bf 42} (1/2) (1998), 101--134.

\bibitem{KrMa} R. Krechetnikov, J. E. Marsden: Dissipation-induced instabilities in finite dimensions, {\sl Reviews of Modern Physics}, {\bf 79} (2), (2007), 519--553.

\bibitem{LeMaMa} M. de Le\'on, J.C. Marrero, D. Mart{\'\i}n de
Diego: Linear almost Poisson structures and Hamilton-Jacobi
theory. Applications to nonholonomic Mechanics, Preprint 2008,
arXiv:0801.4358.

\bibitem{LMM}
 M.~de Le\'{o}n, J.C.~Marrero, E.~Mart\'{\i}nez: Lagrangian submanifolds and dynamics
on Lie algebroids,  {\sl J. Phys. A: Math. Gen.} \textbf{38} (2005), R241--R308.

\bibitem{LR} M. de Le\'on, P.R. Rodrigues: {\sl Methods of Differential Geometry in Analytical
Mechanics},
North Holland Math. Series {\bf 152} (Amsterdam, 1996).

\bibitem{TD69} P. Libermann: Lie algebroids and Mechanics, {\sl Arch. Math. (Brno),} {\bf 32 }(1996), 147-162.

\bibitem{Mac} K. Mackenzie: {\sl General Theory of Lie groupoids and Lie algebroids in
differential geometry}, London Mathematical Society Lecture Note
Series, No. 124, 2005.

\bibitem{MaMaSo} J.C. Marrero, D. Mart\'\i n de Diego, D. Sosa: Variational constrained Mechanics on Lie affgebroids, {\sl Discrete and Continuous dynamical Systems, Series S}, {\bf 3}, (1) (2010) 105-128.


\bibitem{MaSo} J.C.~Marrero, D. Sosa: The Hamilton-Jacobi equation on Lie
affgebroids, {\sl Int. J. Geom. Meth. Mod. Phys.} {\bf 3 } (3)
(2006) 605--622.

\bibitem{Ma} E. Mart\'\i nez: Lagrangian mechanics on Lie algebroids, {\sl Acta Appl. Math. } {\bf 67}  (2001),  no. 3, 295--320.

\bibitem{MaMeSa} E. Mart\'\i nez, T. Mestdag, W. Sarlet: Lie algebroid structures and Lagrangian systems on affine bundles. {\sl J. Geom. Phys.} {\bf 44}  (2002),  no. 1, 70--95.

\bibitem{Me} T. Mestdag: Lagrangian reduction by stages for
    nonholonomic systems in a Lie algebroid framework, {\sl J.  Phys.
    A: Math. Gen.} {\bf 38} (2005), 10157--10179.

\bibitem{MeLa} T. Mestdag, B. Langerock: A Lie algebroid
    framework for nonholonomic systems, {\sl J. Phys. A: Math.
    Gen} {\bf 38} (2005), 1097--1111.

\bibitem{Mont} R. Montgomery: {\sl A Tour of Subriemannian Geometries,
their geodesics and applicaions.} Mathematical Surveys and
Monographs, {\bf 91}. AMS. 2002.

\bibitem{Murray} C.D. Murray: Dynamical effects of drag in the circular
restricted three-body problem, {\sl Icarus} {\bf 112} (1994), 465-484.

\bibitem{lewis} A.D. Lewis:  Simple Mechanical Control Systems with Constraints.
IEE Transactions on Automatic Control, \textbf{45} (8), (2000),
1420-1436.



\bibitem{BT} T. Ohsawa, A.M.  Bloch: Nonholonomic Hamilton-Jacobi equation and Integrability, Preprint 2009, arXiv:0906.3357.

\bibitem{Po2} M. Popescu, P. Popescu: Geometric objects defined by almost Lie
structures, {\sl Proc. Workshop on Lie algebroids and related
topics in Differential Geometry (Warsaw)} {\bf 54} (Warsaw: Banach
Center Publications) (2001), 217---233.


\bibitem{DianaT} D. Sosa: Afgebroides de Lie y Mec\'anica
Geom\'etrica, Dissertation Thesis, available in
http://www.gmcnetwork.org/files/thesis/dsosa.pdf



\bibitem{sussman} H.J. Sussmann: Orbits of families of vector fields and
integrability of distributions, Transactions of the American
Mathematical Society, 180, (1973), 171–188.


\bibitem{TD108} A.J. Van Der Schaft, B. M. Maschke: On the Hamiltonian formulation of non-holonomic mechanics systems, {\sl Rep. Math. Phys., } {\bf 34} (1994), 225--233.

\bibitem{Weinstein99}
  A. Weinstein: Lagrangian Mechanics and Groupoids, {\sl
     in ``Mechanics day" (Waterloo, ON, 1992),
    Fields Institute Communications \textbf{7}, American Mathematical
    Society}, (1996), 207--231.


\end{thebibliography}
\end{document}